\ifx\onecol\undefined
\documentclass[journal]{IEEEtran}
\else 
\documentclass[onecolumn,draftcls,12pt]{IEEEtran}
\fi

\usepackage[utf8]{inputenc}
\usepackage{physics}
\usepackage{amsmath}
\usepackage{amssymb}
\usepackage{subfigure}
\usepackage{graphicx}
\usepackage{graphics}
\usepackage{color}
\usepackage{psfrag}
\usepackage{cite}
\usepackage{balance}
\usepackage{algorithm}
\usepackage{accents}
\usepackage{amsthm}
\usepackage{bm}
\usepackage{url}
\usepackage{algorithmic}
\usepackage[english]{babel}
\usepackage{multirow}
\usepackage{enumerate}
\usepackage{cases}
\usepackage{stfloats}
\usepackage{dsfont}
\usepackage{soul}
\usepackage{amsfonts}
\usepackage{fancyhdr}
\usepackage{hhline}
\usepackage{array}
\usepackage{booktabs}
\usepackage{framed}
\usepackage{float}
\usepackage{soul}
\usepackage{url}
\usepackage{hyperref}

\usepackage{geometry}
\hypersetup{hypertex = true,
	colorlinks = true, %
	linkcolor = black, %
	urlcolor = blue,%
	citecolor = blue} %

\newcommand{\mb}[1]{{  \mathbf  #1}}  

\begin{document}
\newtheorem{theorem}{Theorem}
\newtheorem{acknowledgement}[theorem]{Acknowledgement}
\newtheorem{axiom}[theorem]{Axiom}
\newtheorem{case}[theorem]{Case}
\newtheorem{claim}[theorem]{Claim}
\newtheorem{conclusion}[theorem]{Conclusion}
\newtheorem{condition}[theorem]{Condition}
\newtheorem{conjecture}[theorem]{Conjecture}
\newtheorem{criterion}[theorem]{Criterion}
\newtheorem{definition}{Definition}
\newtheorem{exercise}[theorem]{Exercise}
\newtheorem{lemma}{Lemma}
\newtheorem{corollary}{Corollary}
\newtheorem{notation}[theorem]{Notation}
\newtheorem{problem}[theorem]{Problem}
\newtheorem{proposition}{Proposition}
\newtheorem{solution}[theorem]{Solution}
\newtheorem{summary}[theorem]{Summary}
\newtheorem{assumption}{Assumption}
\newtheorem{example}{\bf Example}
\newtheorem{remark}{\bf Remark}

\newtheorem{thm}{Corollary}[section]
\renewcommand{\thethm}{\arabic{section}.\arabic{thm}}

\def\qed{$\Box$}
\def\QED{\mbox{\phantom{m}}\nolinebreak\hfill$\,\Box$}
\def\proof{\noindent{\emph{Proof:} }}
\def\poof{\noindent{\emph{Sketch of Proof:} }}
\def
\endproof{\hspace*{\fill}~\qed
\par
\endtrivlist\unskip}
\def\endproof{\hspace*{\fill}~\qed\par\endtrivlist\vskip3pt}

\def\E{\mathsf{E}}
\def\eps{\varepsilon}
\def\phi{\varphi}
\def\Lsp{{\boldsymbol L}}
\def\Bsp{{\boldsymbol B}}
\def\lsp{{\boldsymbol\ell}}
\def\Ltsp{{\Lsp^2}}
\def\Lpsp{{\Lsp^p}}
\def\Linsp{{\Lsp^{\infty}}}
\def\LtR{{\Lsp^2(\Rst)}}
\def\ltZ{{\lsp^2(\Zst)}}
\def\ltsp{{\lsp^2}}
\def\ltZt{{\lsp^2(\Zst^{2})}}
\def\ninN{{n{\in}\Nst}}
\def\oh{{\frac{1}{2}}}
\def\grass{{\cal G}}
\def\ord{{\cal O}}
\def\dist{{d_G}}
\def\conj#1{{\overline#1}}
\def\ntoinf{{n \rightarrow \infty}}
\def\toinf{{\rightarrow \infty}}
\def\tozero{{\rightarrow 0}}
\def\trace{{\operatorname{trace}}}
\def\ord{{\cal O}}
\def\UU{{\cal U}}
\def\rank{{\operatorname{rank}}}
\def\acos{{\operatorname{acos}}}

\def\SINR{\mathsf{SINR}}
\def\SNR{\mathsf{SNR}}
\def\SIR{\mathsf{SIR}}
\def\tSIR{\widetilde{\mathsf{SIR}}}
\def\Ei{\mathsf{Ei}}
\def\l{\left}
\def\r{\right}
\def\lb{\left\{}
\def\rb{\right\}}

\setcounter{page}{1}

\newcommand{\eref}[1]{(\ref{#1})}
\newcommand{\fig}[1]{Fig.\ \ref{#1}}

\def\bydef{:=}
\def\ba{{\mathbf{a}}}
\def\bb{{\mathbf{b}}}
\def\bc{{\mathbf{c}}}
\def\bd{{\mathbf{d}}}
\def\bee{{\mathbf{e}}}
\def\bff{{\mathbf{f}}}
\def\bg{{\mathbf{g}}}
\def\bh{{\mathbf{h}}}
\def\bi{{\mathbf{i}}}
\def\bj{{\mathbf{j}}}
\def\bk{{\mathbf{k}}}
\def\bl{{\mathbf{l}}}
\def\bm{{\mathbf{m}}}
\def\bn{{\mathbf{n}}}
\def\bo{{\mathbf{o}}}
\def\bp{{\mathbf{p}}}
\def\bq{{\mathbf{q}}}
\def\br{{\mathbf{r}}}
\def\bs{{\mathbf{s}}}
\def\bt{{\mathbf{t}}}
\def\bu{{\mathbf{u}}}
\def\bv{{\mathbf{v}}}
\def\bw{{\mathbf{w}}}
\def\bx{{\mathbf{x}}}
\def\by{{\mathbf{y}}}
\def\bz{{\mathbf{z}}}
\def\b0{{\mathbf{0}}}

\def\bA{{\mathbf{A}}}
\def\bB{{\mathbf{B}}}
\def\bC{{\mathbf{C}}}
\def\bD{{\mathbf{D}}}
\def\bE{{\mathbf{E}}}
\def\bF{{\mathbf{F}}}
\def\bG{{\mathbf{G}}}
\def\bH{{\mathbf{H}}}
\def\bI{{\mathbf{I}}}
\def\bJ{{\mathbf{J}}}
\def\bK{{\mathbf{K}}}
\def\bL{{\mathbf{L}}}
\def\bM{{\mathbf{M}}}
\def\bN{{\mathbf{N}}}
\def\bO{{\mathbf{O}}}
\def\bP{{\mathbf{P}}}
\def\bQ{{\mathbf{Q}}}
\def\bR{{\mathbf{R}}}
\def\bS{{\mathbf{S}}}
\def\bT{{\mathbf{T}}}
\def\bU{{\mathbf{U}}}
\def\bV{{\mathbf{V}}}
\def\bW{{\mathbf{W}}}
\def\bX{{\mathbf{X}}}
\def\bY{{\mathbf{Y}}}
\def\bZ{{\mathbf{Z}}}

\def\bxi{{\boldsymbol{\xi}}}

\def\sT{{\mathsf{T}}}
\def\sH{{\mathsf{H}}}
\def\cmp{{\text{cmp}}}
\def\cmm{{\text{cmm}}}
\def\WPT{{\text{WPT}}}
\def\lo{{\text{lo}}}
\def\gl{{\text{gl}}}

\def\tT{{\widetilde{T}}}
\def\tF{{\widetilde{F}}}
\def\tP{{\widetilde{P}}}
\def\tG{{\widetilde{G}}}
\def\tbh{{\widetilde{\mathbf{h}}}}
\def\tbg{{\widetilde{\mathbf{g}}}}

\def\mA{{\mathbb{A}}}
\def\mB{{\mathbb{B}}}
\def\mC{{\mathbb{C}}}
\def\mD{{\mathbb{D}}}
\def\mE{{\mathbb{E}}}
\def\mF{{\mathbb{F}}}
\def\mG{{\mathbb{G}}}
\def\mH{{\mathbb{H}}}
\def\mI{{\mathbb{I}}}
\def\mJ{{\mathbb{J}}}
\def\mK{{\mathbb{K}}}
\def\mL{{\mathbb{L}}}
\def\mM{{\mathbb{M}}}
\def\mN{{\mathbb{N}}}
\def\mO{{\mathbb{O}}}
\def\mP{{\mathbb{P}}}
\def\mQ{{\mathbb{Q}}}
\def\mR{{\mathbb{R}}}
\def\mS{{\mathbb{S}}}
\def\mT{{\mathbb{T}}}
\def\mU{{\mathbb{U}}}
\def\mV{{\mathbb{V}}}
\def\mW{{\mathbb{W}}}
\def\mX{{\mathbb{X}}}
\def\mY{{\mathbb{Y}}}
\def\mZ{{\mathbb{Z}}}

\def\cA{\mathcal{A}}
\def\cB{\mathcal{B}}
\def\cC{\mathcal{C}}
\def\cD{\mathcal{D}}
\def\cE{\mathcal{E}}
\def\cF{\mathcal{F}}
\def\cG{\mathcal{G}}
\def\cH{\mathcal{H}}
\def\cI{\mathcal{I}}
\def\cJ{\mathcal{J}}
\def\cK{\mathcal{K}}
\def\cL{\mathcal{L}}
\def\cM{\mathcal{M}}
\def\cN{\mathcal{N}}
\def\cO{\mathcal{O}}
\def\cP{\mathcal{P}}
\def\cQ{\mathcal{Q}}
\def\cR{\mathcal{R}}
\def\cS{\mathcal{S}}
\def\cT{\mathcal{T}}
\def\cU{\mathcal{U}}
\def\cV{\mathcal{V}}
\def\cW{\mathcal{W}}
\def\cX{\mathcal{X}}
\def\cY{\mathcal{Y}}
\def\cZ{\mathcal{Z}}
\def\cd{\mathcal{d}}
\def\Mt{M_{t}}
\def\Mr{M_{r}}
\def\O{\Omega_{M_{t}}}
\newcommand{\figref}[1]{{Fig.}~\ref{#1}}
\newcommand{\tabref}[1]{{Table}~\ref{#1}}

\newcommand{\fb}{\tx{fb}}
\newcommand{\nf}{\tx{nf}}
\newcommand{\BC}{\tx{(bc)}}
\newcommand{\MAC}{\tx{(mac)}}
\newcommand{\Pout}{p_{\mathsf{out}}}
\newcommand{\nnn}{\nn\\}
\newcommand{\FB}{\tx{FB}}
\newcommand{\TX}{\tx{TX}}
\newcommand{\RX}{\tx{RX}}
\renewcommand{\mod}{\tx{mod}}
\newcommand{\m}[1]{\mathbf{#1}}
\newcommand{\td}[1]{\tilde{#1}}
\newcommand{\sbf}[1]{\scriptsize{\textbf{#1}}}
\newcommand{\stxt}[1]{\scriptsize{\textrm{#1}}}
\newcommand{\suml}[2]{\sum\limits_{#1}^{#2}}
\newcommand{\sumlk}{\sum\limits_{k=0}^{K-1}}
\newcommand{\eqhsp}{\hspace{10 pt}}
\newcommand{\tx}[1]{\texttt{#1}}
\newcommand{\Hz}{\ \tx{Hz}}
\newcommand{\sinc}{\tx{sinc}}
\newcommand{\diag}{\mathrm{diag}}
\newcommand{\MAI}{\tx{MAI}}
\newcommand{\ISI}{\tx{ISI}}
\newcommand{\IBI}{\tx{IBI}}
\newcommand{\CN}{\tx{CN}}
\newcommand{\CP}{\tx{CP}}
\newcommand{\ZP}{\tx{ZP}}
\newcommand{\ZF}{\tx{ZF}}
\newcommand{\SP}{\tx{SP}}
\newcommand{\MMSE}{\tx{MMSE}}
\newcommand{\MINF}{\tx{MINF}}
\newcommand{\RC}{\tx{MP}}
\newcommand{\MBER}{\tx{MBER}}
\newcommand{\MSNR}{\tx{MSNR}}
\newcommand{\MCAP}{\tx{MCAP}}
\newcommand{\vol}{\tx{vol}}
\newcommand{\ah}{\hat{g}}
\newcommand{\tg}{\tilde{g}}
\newcommand{\teta}{\tilde{\eta}}
\newcommand{\heta}{\hat{\eta}}
\newcommand{\uh}{\m{\hat{s}}}
\newcommand{\eh}{\m{\hat{\eta}}}
\newcommand{\hv}{\m{h}}
\newcommand{\hh}{\m{\hat{h}}}
\newcommand{\Po}{P_{\mathrm{out}}}
\newcommand{\Poh}{\hat{P}_{\mathrm{out}}}
\newcommand{\Ph}{\hat{\gamma}}
\newcommand{\mat}[1]{\begin{matrix}#1\end{matrix}}
\newcommand{\ud}{^{\dagger}}
\newcommand{\C}{\mathcal{C}}
\newcommand{\nn}{\nonumber}
\newcommand{\nInf}{U\rightarrow \infty}

\title{Towards Atomic MIMO Receivers}
\author{{Mingyao Cui,~\IEEEmembership{Graduate Student Member, IEEE}, Qunsong Zeng,~\IEEEmembership{Member, IEEE}, and Kaibin Huang,~\IEEEmembership{Fellow, IEEE}}
\thanks{The authors are with the Department of Electrical and Electronic Engineering, The University of Hong Kong, Hong Kong. Corresponding authors: K. Huang (Email: huangkb@eee.hku.hk); Q. Zeng (Email: qszeng@eee.hku.hk).}}

\maketitle

\thispagestyle{empty}
\pagestyle{empty}

\begin{abstract}
The advancement of Rydberg atoms in quantum information technology is driving a paradigm shift from classical \emph{radio-frequency} (RF) receivers to atomic receivers. 
Capitalizing on the extreme sensitivity of Rydberg atoms to external 
electromagnetic fields,  atomic receivers are capable of realizing more precise 
radio-wave measurements than RF receivers to support high-performance wireless 
communication and 
sensing. 
Although the atomic receiver is developing rapidly in 
quantum-physics domain, its integration with wireless communications is at a 
nascent stage. In particular, systematic methods to enhance communication 
performance through this integration are yet to be discovered.
Motivated by this observation, we propose in this paper to incorporate atomic 
receivers into \emph{multiple-input-multiple-output} (MIMO) communication, a 
prominent 5G technology, as the first attempt on implementing atomic MIMO 
receivers.
To begin with, we provide a comprehensive introduction on the principles of 
atomic receivers and build on them to design the atomic MIMO receivers.
Our findings reveal that signal detection of atomic MIMO receivers corresponds 
to a non-linear biased \emph{phase retrieval} (PR) problem, as opposed to the 
linear Gaussian model adopted in classical MIMO systems. 
Then, to recover signals from this non-linear model, we modify and
integrate the Gerchberg-Saxton (GS) algorithm, a typical PR solver, with a 
biased GS algorithm to solve the biased PR problem. 
Moreover, we propose a novel Expectation-Maximization GS (EM-GS) algorithm to cope with the unique Rician distribution of the biased PR model. 
Our EM-GS algorithm introduces a high-pass filter constructed by the ratio of 
Bessel functions into the iteration procedure of GS, thereby improving the 
detection accuracy without sacrificing the computational efficiency. 
Finally, the effectiveness of the devised algorithms and the feasibility of atomic MIMO receivers are demonstrated by theoretical analysis and numerical simulation. 




\end{abstract}


\section{Introduction}\label{sec:1}
With the global deployment of commercial 5G services, cutting-edge research on 
6G has been launched worldwide~\cite{6GReview_Zhang2019}. Meanwhile, the 
revolutionary \emph{quantum information technologies} are evolving 
rapidly and deemed as a key enabler of the next-generation 
computing, communication, and sensing~\cite{QuanComp_Gyongyosi2019, 
QuanComm_Wang2022, QuanSense_Degen2017}. 
For example, ever since the invention of Schor's algorithm, quantum computing 
has exhibited unparalleled abilities in quadratically or even exponentially 
accelerating the computing speed in solving classically intractable problems, 
ranging from factorization to optimization to artificial intelligence~\cite{QuanComp_Gyongyosi2019}.
On the other hand, leveraging quantum mechanisms such as superposition and 
entanglement, quantum communication is capable of establishing unconditionally 
secure links to protect data privacy~\cite{QuanComm_Wang2022}. 
In parallel with communication and computing, quantum sensing is another 
promising branch of quantum information technologies~\cite{QuanSense_Degen2017}. It 
employs microscopic particles as ``sensors" and capitalizes on their strong 
sensitivity to external disturbance, especially electromagnetic fields, to 
precisely measure physical quantities. Notable examples of quantum sensors 
include 
atomic clocks and superconducting interferometers. As an exciting trend,  
quantum sensors are being applied to an increasingly broad range of practical 
use cases~\cite{QuanSense_Degen2017, QuanSense_Zhang2023, RydMag_Fancher2021}.
Relevant techniques promise to reshape the information society and our 
perception of the environment. Aligned with this end, our work advocates the 
use of quantum sensors to revolutionize next-generation wireless receivers.  



\subsection{Concept and Properties of Atomic Receiver}
Recently, at the intersection of quantum sensing and wireless communications, 
an emerging concept known as atomic receiver has attracted extensive 
attentions, for their capabilities of high-precision measurement of 
electromagnetic 
fields~\cite{RydMag_Fancher2021, liu_continuous-frequency_2022, liao_microwave_2020}. 
An atomic receiver leverages \emph{Rydberg atoms} as ``antennas" to detect the amplitude, phase, frequency, and polarization of incident electromagnetic waves~\cite{RydMag_Liu2023, zhong_polar_2023, liu_deep_2022}. 
To elaborate, Rydberg atoms refer to highly excited atoms, whose outermost 
electrons are excited to high energy levels far from their nuclei~\cite{RydReview_Saffman2010}. 
Attributed to their large electric dipole moments and high polarizability, 
such atoms are extremely sensitive to external electromagnetic fields, 
rendering them excellent ``atomic antennas". 
To implement wireless receivers, a modulated electromagnetic wave first 
interacts with Rydberg atoms to trigger their electronic transitions between 
certain energy levels. 
The resultant intensity of electron transition can be monitored using an 
all-optical readout mechanism, termed the \emph{electromagnetically induced 
transparency} (EIT), to facilitate the ensuing information demodulation~\cite{RydMag_Fancher2021}.  
In contrast with the macroscopic measurement of electromagnetic fields by 
conventional \emph{radio-frequency} (RF) receivers, the atomic 
counterparts operate at the atom-scale microcosm. 
This fundamentally different principle enables the latter to achieve accuracies 
at the quantum limit and have the potential of working in tandem with or even 
replacing conventional RF receivers~\cite{RydMag_Art2022}. 

Benefiting from the advancements in quantum sensing, atomic receivers offer 
several advantages over conventional receivers for wireless communication and 
sensing.  
First, atomic receivers can achieve much higher detection-and-sensing 
accuracies as mentioned earlier~\cite{QuanSense_Degen2017}. 
For example, as demonstrated in~\cite{QuanSense_Zhang2023}, atomic receivers 
can achieve the sub-millimeter (sub-mmWave) level granularity in detecting 
object vibrations using WiFi sources as opposed to the mmWave level granularity 
attained by traditional receivers.
The second advantage lies in low thermal noise. 
Conventional receivers require RF circuits to intercept radio waves and 
downconvert RF signals to the baseband, which comprise components such as analog amplifiers and mixers that generate significant thermal noise.
In contrast, atomic receivers employ all-optical components for sensing and 
downconversion, and thus are immune to thermal noise~\cite{RydMag_Liu2023}.
The above two advantages make atomic receivers a promising technology for 
long-range 
communications~\cite{RydMag_Fancher2021}, such as satellite communications and space-air-ground networks~\cite{liu_space-air-ground_2018}, where detecting weak electric signals is crucial.  
Last but not the least, the flexible tunability of Rydberg atoms across an 
ultra-wide spectrum is a significant advantage~\cite{RydMag_Art2022}.
Traditional antennas, with a fixed size typically half the wavelength of 
the carrier, are limited in the spectrum bandwidth they can listen to. 
Detecting 
multi-band signals requires an array of costly RF receivers, each of which is 
matched to a specific band~\cite{Multiband_Sim2020}.
In contrast, the numerous energy levels in Rydberg atoms allow an atomic 
receiver to respond to a broad range of frequencies from 100 Megahertz (MHz) to 
1 Terahertz (THz) by exciting electrons to different energy levels~\cite{RydMultiband_Du2022}. { This capability enables the simultaneous 
measurement of multi-band signals over an extremely wide spectrum using a 
single atomic receiver, laying the foundation for constructing an 
extraordinarily cost-effective 
full-frequency communication platform.
To materialize the above visions, this work incorporates atomic receivers into 
modern wireless communication systems.}

\subsection{State-of-the-art Atomic Receivers}
Existing research on atomic receivers primarily focuses on the lab demonstrations of different signal modulation schemes. 
It is well-known that via the interaction between electromagnetic fields and Rydberg atoms, the strength and frequency of modulated radio waves can be transduced into the intensity of electron transition, called the \emph{Rabi frequency}. 
Hence, by inferring symbols from the Rabi frequencies, both the \emph{amplitude demodulation} and \emph{frequency demodulation} at $37.4$ GHz and $29.5$ GHz have been realized, respectively~\cite{RydAMFM_Anderson2021}. 
For \emph{phase demodulation}, researchers employed the \emph{holographic phase-sensing} methodology, involving the use of a \emph{local oscillator} (LO) to generate a reference electromagnetic field with a known phase that interferes with the incident modulated electromagnetic field~\cite{RydPhase_Simons2019, RydNP_Jing2020, RydPhase_And2020}. Then, the phase-modulated symbols can be extracted from the phase differences between the aforementioned fields. In~\cite{RydPhase_Meyer2018}, a channel capacity of up to 8.2 Mbps was measured using the 8-state phase-shift-keying modulation. This holographic principle was then extended to support \emph{angle-of-arrival} estimation~\cite{RydAOA_Robinson2021} and micro-vibration monitoring~\cite{QuanSense_Zhang2023}. Another vein of research exploits the capability of atomic receivers for multi-band communications~\cite{RydMultiband_Du2022, RydMultiband_Meyer2023}. For example, it was proposed in~\cite{RydMultiband_Meyer2023}  to excite electrons to the same initial Rydberg energy level but transfer them to different final Rydberg energy levels via the interaction with separate frequency bands. This approach enables the co-detection of 5 frequency bands at 1.72, 12.11, 27.42, 65.11, and 115.75 GHz. 

Despite the rapid development of atomic receivers in the quantum-sensing domain, their application in wireless communication is still at a nascent stage. 
The state-of-the-art designs are suitable only for primitive communication scenarios and schemes, such as a point-to-point system or a \emph{single-input-single-output} (SISO) air interface. 
Without advanced designs and advices, the full potential of atomic receivers cannot be unleashed, with the state-of-the-art achieving only a Mbps-scale channel capacity~\cite{RydPhase_Meyer2018, RydMultiband_Du2022, RydMultiband_Meyer2023}. 
Communication techniques that are customized for atomic receivers are largely uncharted. This calls for research on atomic-level channel coding, massive \emph{multiple-input-multiple-output} (MIMO), multi-user communications, \emph{integrated sensing and communication} (ISAC), \emph{reconfigurable intelligent surface} (RIS), and edge learning~\cite{MIMO_LU2014, RIS_Zhang2023,EdgeAI_Mao2017}. 
{ To contribute to this emerging field}, this paper seeks the feasibility of integrating MIMO communications and atomic receivers. 
 MIMO, a key 5G technology, features spatial multiplexing of parallel data links and enhancement of link reliability via diversity gain~\cite{MIMO_LU2014}. 
It is believed that the consideration of MIMO has the potential to overcome the capacity and sensitivity bottleneck of current atomic SISO receivers.

\subsection{Contributions and Organization}
To the best of our knowledge, this work represents the first attempt on designing and analyzing atomic MIMO receivers. We propose a novel framework for atomic MIMO receiver by building on the basic mechanisms of quantum sensing. 
    In this framework, phase-modulated symbols simultaneously transmitted by multiple users are detected by multiple atomic antennas using the mentioned scheme of holographic phase-sensing. 
The key contributions and findings are summarized as follows. 

\begin{itemize}
    \item \textbf{Modeling of atomic MIMO receiver:} 
    Our model of atomic MIMO receiver reveals that its signal detection is essentially a \emph{biased phase retrieval (PR) problem}. 
    In particular, the transmitted symbols are modulated in the coupling strength between the electric dipole moment of Rydberg atoms and the incident radio waves from the users and the reference source.
    The atomic MIMO receiver capitalizes on this strength to infer multi-user symbols without the help of phase information.  
    This finding highlights one crucial departure of atomic MIMO system from its traditional counterpart: as opposed to the linear transmission model fitting the latter, the former operates in a \emph{non-linear transmission model}. 
    \item \textbf{Signal detection for atomic MIMO receiver:}
    To solve the discovered biased PR problem, we propose two signal detection algorithms: the biased Gerchberg-Saxton (GS) and the Expectation-Maximization GS (EM-GS) algorithms, which are based on the \emph{least square} (LS) and \emph{maximum likelihood} (ML) criteria, respectively. 
    First, the biased GS algorithm extends the classic GS algorithm~\cite{PR_Netrapalli2015}, a popular PR solver, by eliminating the bias caused by the reference source in the initialization and iteration steps of the classic algorithm.
    Next, the proposed EM-GS algorithm employs the Bayesian regression to perform ML detection. Its novelty lies in treating the unobserved phase information as a latent variable and thereby decoupling the intricate ML problem into a sequence of tractable linear regression problems with analytical solutions. 
    A comparison of the proposed EM-GS algorithm to the biased GS algorithm reveals that the former features an additional high-pass filter as constructed by the ratio of Bessel functions. This entails the improvement of detection accuracy without increasing computational complexity.
    \item \textbf{Performance study:}
    We employ two metrics, namely the \emph{normalized mean square error} (NMSE) and the \emph{bit error ratio} (BER), to validate the feasibility of atomic MIMO receivers as well as the effectiveness of the proposed algorithms. 
    Targeting the NMSE metric, the \emph{Cramér-Rao lower bound} (CRLB) is derived. 
    Numerical results demonstrate that the proposed EM-GS algorithm consistently approaches the CRLB and outperforms the biased GS algorithm, especially in the low \emph{signal-to-noise ratio} (SNR) regime. 
    Regarding the BER performance, simulation results confirm that the biased GS and EM-GS algorithms achieve BERs close to those from an exhaustive search,  while being much more computationally efficient.
\end{itemize}

The remainder of this paper is organized as follows. The preliminaries of  atomic SISO receivers are provided in Section~\ref{sec:2}. The framework of atomic MIMO receiver and the corresponding detection algorithms are presented in Section~\ref{sec:3} and Section~\ref{sec:4}, respectively.  
Section V extends the proposed designs to channel estimation problems.
In Section~\ref{sec:5}, we validate the proposed design by numerical results. Finally, conclusions are drawn in Section~\ref{sec:6}. 





\ifx\onecol\undefined
\begin{figure*}
   \centering
   \includegraphics[width=6in]{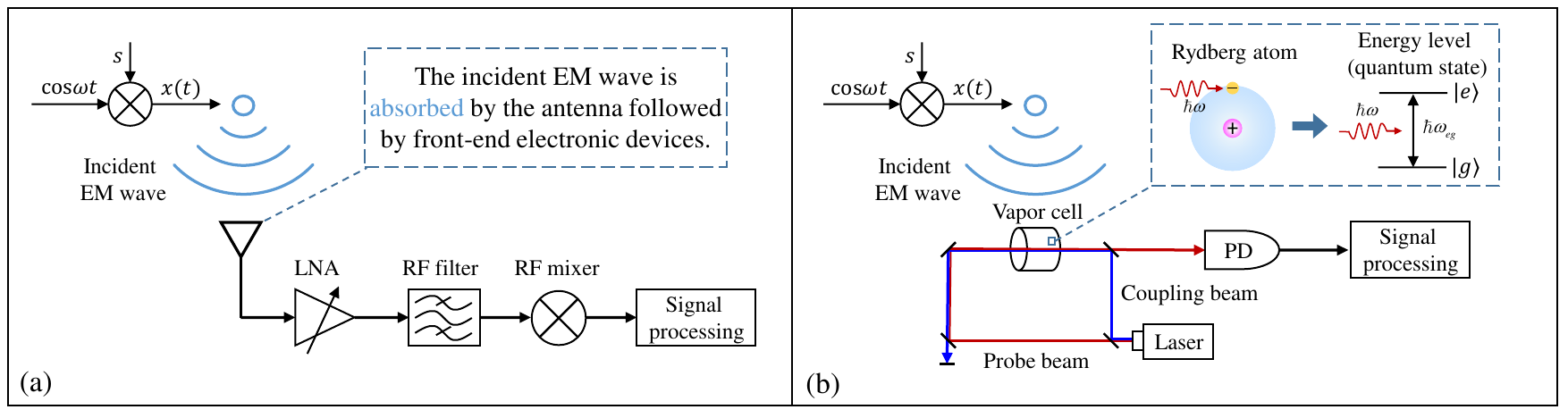}
   \vspace*{-1em}
   \caption{\color{black} Wireless communication systems based on (a) classical 
   receiver and (b) atomic receiver. 
   }
   \label{img:Receivers}
   \vspace*{-1em}
\end{figure*}
\else 
\begin{figure*}
   \centering
   \includegraphics[width=6in]{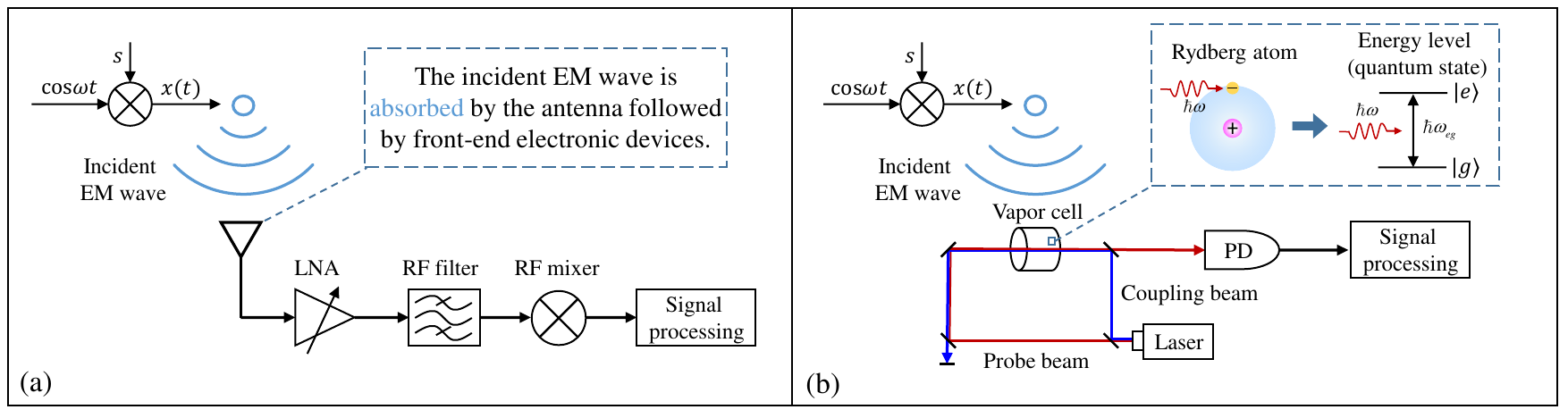}
   \caption{Wireless communication systems based on (a) classical receiver and (b) atomic receiver. (c) Rydberg atoms and energy level diagram.
   }
   \label{img:Receivers}
\end{figure*}
\fi 

\section{Preliminaries of Atomic Wireless Receiver}\label{sec:2}
To help understanding of the proposed framework, preliminaries of atomic receiver for a SISO air interface is provided in this section. We first present the architecture of an atomic receiver and explain the underlying physical laws of Rydberg atoms as antennas, followed by an introduction to quantum jump enabling wireless signal detection. Last, we end this section with the implementation of atomic receivers.

\subsection{Atomic Receiver Architecture}\label{sec:2A}
To begin with, the wireless communication systems employing classical and atomic receivers are differentiated in Fig.~\ref{img:Receivers}. 
The two systems are identical in terms of transmitter architecture and channel. Specifically, data symbols are pre-processed using electronic circuits and modulated onto electromagnetic waves. 
Their distinction lies in the receiver architecture as elaborated in the sequel.
A classical receiver utilizes a dipole antenna to transduce incident  waves into signals and decodes from them transmitted bits using a series of electronic components, such as mixers, amplifiers, filters, and baseband professors.
In contrast, an atomic receiver employs a different signal-detection paradigm. 
Such a receiver deploys Rydberg atoms as its ``antennas" to measure the amplitude, phase, frequency, and polarization of an electromagnetic field, leveraging their exceptional sensitivity to electromagnetic waves ranging from MHz to THz. 
Free from energy-consuming and bulky electronic components, a vapor cell filled with Rydberg atoms suffices for sensing the wireless environment. 
The underpinning mechanism involves the incident electromagnetic wave interacting with Rydberg atoms to change their quantum states, which are monitored using an all-optical readout device, such as a \emph{photodetector} (PD). The output is processed by a signal-detection algorithm to decipher the transmitted bits.
In summary, the architecture of an atomic-receiver based communication system comprises a classic transmitter, a classic channel, and a quantum receiver. 

\subsection{Rydberg Atoms as Antennas}\label{sec:2B}
\subsubsection{Atomic Energy Levels}
To grasp the concept of atomic receiver, it is essential to first understand the principle of Rydberg atoms. 
Each atom consists of a positively charged nucleus surrounded by multiple negatively charged electrons. 
According to the Rutherford–Bohr model, the Coulomb potential from the nucleus confines the electrons to discrete orbits with distinct energy levels~\cite{AtomicPhysics}. 
An energy level can be denoted as $n\ell_j$, where $n$, $\ell$, and $j$ represent the principal quantum number, the orbit angular momentum quantum number, and the total angular momentum quantum number, respectively~\cite{QuantumMechanism_Zwiebach2006}.
While $n$ and $j$ are numerical, $\ell$ is assigned a letter from $S,P,D,F,\cdots,$ in an ascending order according to the spectroscopic notation.
For instance, $n\ell_j = 50P_{1/2}$.
Let $E_{n\ell j}$ denote the energy corresponding to the energy level $n\ell_j$. In particular, for a hydrogen-like atom, 
\begin{align} \label{eq:EnergyLevel}
    E_{n\ell j} = -\frac{hc{\rm Ry}}{(n-\delta_{nlj})^2},
\end{align}
where $h=6.626\times10^{-34}$ J$\cdot$s is the Planck constant, $c=2.998\times10^8$ m$\cdot$s$^{-1}$ the speed of light, ${\rm Ry} = 1.079\times 10^{7}$~m${^{-1}}$ the Rydberg constant, and $\delta_{nlj}$ the quantum defect constant specific for each energy level~\cite{RydMag_Fancher2021}.
For ease of notation, we omit the subscript $\ell j$ of $E_{n\ell j}$ in the sequel. 
Besides, it is useful to define the \emph{intrinsic frequency} of each energy level as $\omega_n \overset{\Delta}{=} \frac{E_n}{\hbar}$, and the \emph{transition frequency} between two energy levels as $\omega_{nm} \overset{\Delta}{=} \omega_n - \omega_m = \frac{E_n - E_m}{\hbar}$, where $\hbar = \frac{h}{2\pi}$ is the reduced Planck constant.  

\subsubsection{Quantum Jump}
Electrons can jump between energy levels by either absorbing or emitting photons~\cite{QuantumOptics_Fox2006}. 
When a photon is absorbed, the electron is excited to a higher energy level, while transitioning to a lower energy level results in photon emission. 
 This peculiar quantum phenomenon is known as quantum jump or electron transition. 
It is worth mentioning that the electromagnetic radiation/sensing essential for wireless communication is a special form of photon emission/absorption due to the wave–particle duality~\cite{QuantumMechanism_Zwiebach2006}. 
The energy conservation law dictates that an electron transition between two energy levels $E_n$ and $E_m$ ($E_n > E_m$) can occur only when the angular frequency $\omega$ of electromagnetic wave is close to the transition frequency, $\omega_{nm}$, (i.e., near-resonance or resonance). 
The quantum jumps in response to carrier frequencies are the fundamental phenomenon enabling atomic receivers.

\subsubsection{Definition of Rydberg Atom}
An atom becomes a Rydberg atom when at least one of its electrons is excited to a high energy level (Rydberg state) with a principal quantum number $n > 20$~\cite{RydReview_Saffman2010}.
As depicted in Fig.~\ref{img:Receivers}(b), to prepare Rydberg atoms, two laser beams of different frequencies and intensities (i.e., a weak probe beam and a strong coupling beam)  
are steered in opposite directions to propagate through a vapor cell filled with alkali-metal atoms~\cite{RydAMFM_Anderson2021}. 
The probe beam excites the outermost electrons from their ground level to a low excited state (e.g., $6S_{1/2}\rightarrow6P_{1/2}$). 
The coupling beam then further triggers these electrons to jump from the state to a highly excited Rydberg state (e.g., $6P_{1/2}\rightarrow47S_{1/2}$), thereby creating Rydberg atoms. 
Apart from the preparation of Rydberg states, the probe beam also serves the purpose of reading out the quantum state of Rydberg atoms, as elaborated in Section \ref{sec:2D}. 

\subsubsection{Properties of Rydberg Atom}
Rydberg atoms exhibit several properties useful for wireless communication due to the large electron-nucleus separation distance. 
First, Bohr's model indicates that the orbit radius of an electron is proportional to $n^2$~\cite{RydReview_Saffman2010}. 
The large principal quantum number in a Rydberg atom results in a large atomic radius and, consequently, makes Rydberg atoms extremely sensitive to incident electromagnetic waves, enabling accurate signal detection.
Second, it follows from \eqref{eq:EnergyLevel} that the energy spacing between two adjacent energy levels is approximately $\Delta E \propto\frac{1}{n^3}$. 
This scaling reveals that energy levels become denser with increasing $n$. 
As a result, the set of available transition frequencies $\omega_{nm}$ becomes larger to cover a wide range of spectrum from MHz to THz~\cite{QuanSense_Zhang2023}. 
For example, the transition between energy levels $52D_{5/2}$ and $53P_{3/2}$ corresponds to a transition frequency $\omega_{nm} = 2\pi \times 5$ GHz, which falls within the sub-6G band. 
On the other hand, mmWave signals at $28$ GHz can trigger the transition between energy levels $61D_{5/2}$ and $63P_{1/2}$. 
{ The preceding advantages inspire researchers to develop atomic receivers based on Rydberg atoms to realize wireless detection across a vast range of communication bands~\cite{RydMultiband_Du2022, RydMultiband_Meyer2023}.}

\subsection{Electromagnetic Wave and Quantum Jump}\label{sec:2C}

\subsubsection{Quantum State of Rydberg Atom}
Let $P$ and $s$ denote transmit power and baseband signal in a narrowband system, respectively. 
Consider the system in Fig.~\ref{img:Receivers}(b), the incident electromagnetic wave at the receiver is given as 
\begin{align} \label{eq:RFSignal}
    \mb{E}(t) = \boldsymbol{\epsilon } \sqrt{P}\rho s \cos (\omega t + \phi) = [E_x(t),E_y(t), E_z(t)]^T,
\end{align}
where $\rho$ represents the path loss, $\phi$ the phase shift, and the unit-vector, $\boldsymbol{\epsilon} = [\epsilon_x, \epsilon_y, \epsilon_z]^T$ with $\|\boldsymbol{\epsilon}\|_2 = 1$, the polarization direction. 
The receiver harnesses a vapor cell filled with Rydberg atoms to detect the wave by observing quantum jumps. 
Specifically, we denote the energy levels of the Rydberg ground and excited states as $E_g = \hbar \omega_g$ and $E_e = \hbar \omega_e$. The frequency $\omega$ of electromagnetic wave is set close to the transition frequency $\omega_{eg} = \omega_e - \omega_g$ such that quantum jumps are primarily between these two energy levels. 
Then, the quantum state of this two-level system is expressed by a $2 \times 1$ complex vector\footnote{In quantum physics, the Dirac notation is widely used for quantum states, where a ``bra'' represents the conjugate transpose of a ket vector, denoted as $\bra{\psi}$, while a ``ket'' represents a vector in Hilbert space, denoted as $\ket{\psi}$. The inner product of two kets $\ket{\psi}$ and  $\ket{\phi}$ is expressed as $\bra{\psi}\ket{\phi}$~\cite{QuantumMechanism_Zwiebach2006}.}, 
whose ground and excited states are given as $\ket{g} = [0, 1]^T$ and $\ket{e} = [1, 0]^T$, respectively.
A general quantum state, $\ket{\psi}$, is  modeled by the linear combination of $\ket{g}$ and $\ket{e}$:
\begin{align} \label{eq:State}
    \ket{\psi} = \alpha_g \ket{g} + \alpha_e \ket{e} = [\alpha_e, \alpha_g]^T,
\end{align}
where the complex coefficients, $\alpha_g$ and $\alpha_e$, are called the probability amplitudes of $\ket{\psi}$, satisfying $|\alpha_g|^2 + |\alpha_e|^2 = 1$. 
When $\ket{\psi}$ is measured on the basis $\{\ket{g}, \ket{e}\}$, the outcome is $\ket{g}$ with probability $|\alpha_g|^2$ and $\ket{e}$ with probability $|\alpha_e|^2$.
\subsubsection{State Evolution of Rydberg Atom}
For a time-varying system we consider, i.e., $\ket{\psi(t)} = \alpha_g(t) \ket{g} + \alpha_e(t) \ket{e}$, the evolution of quantum state is governed by the Schrödinger equation~\cite{QuantumMechanism_Zwiebach2006}:
\begin{align}\label{eq:Schrodinger}
    i\hbar \frac{\partial \ket{\psi(t)}}{\partial t} = (\hat{H} + \hat{V})\ket{\psi(t)}, 
\end{align}
where $i = \sqrt{-1}$. The operators $\hat{H}$ and $\hat{V}$ are known as the free Hamiltonian and interaction Hamiltonian operators, which are defined as follows. 
\begin{itemize}
    \item \textbf{Free Hamiltonian}: The free Hamiltonian $\hat{H}$ represents the energy of the isolated nucleus-electron system. 
By acting $\hat{H}$ on the ground state $\ket{g}$ and the excited state $\ket{e}$, we can obtain their energy levels $E_g = \hbar \omega_g$ and $E_e = \hbar \omega_e$.
In this context, $\hat{H}$ is expressed by a $2\times 2$ diagonal matrix:
\ifx\onecol\undefined
\begin{align}\label{eq:FreeHamiltonian}
    \hat{H} = \hbar\omega_e \ket{e}\bra{e} + \hbar\omega_g \ket{g}\bra{g} = \diag\{\hbar \omega_e, \hbar \omega_g\}.
\end{align}
\else 
\begin{align}\label{eq:FreeHamiltonian}
    \hat{H} = \hbar\omega_e \ket{e}\bra{e} + \hbar\omega_g \ket{g}\bra{g} = \diag\{\hbar \omega_e, \hbar \omega_g\}.
\end{align}
\fi

\item\textbf{Interaction Hamiltonian}: The Hamiltonian $\hat{V}$ accounts for the interaction energy of the incident electromagnetic field \eqref{eq:RFSignal} acting on this two-level system. 
Using Rutherford–Bohr model, we define the relative position from the nucleus to the electron as $\mb{r}$. 
Then, the interaction energy in classical physics is determined by $V = q\mb{r}^T\mb{E}(t) = \boldsymbol{\mu}^T\mb{E}(t)$, where $q$ is the charge of an electron and $\boldsymbol{\mu} = q\mb{r}$ the electric dipole moment. 
In quantum physics, the position vector $\mb{r}$ is translated to a position operator  $\hat{\mb{r}} = [\hat{r}_x, \hat{r}_y, \hat{r}_z]^T$. 
Each entry of $\hat{\mb{r}}$ refers to a quantum operator denoted by a $2\times 2$ Hermitian matrix. 
According to the \emph{selection role} of electron transition~\cite{QuantumOptics_Fox2006}, the expansion of $\hat{r}_d$, $d\in\{x,y,z\}$ on the basis $\{\ket{e}, \ket{g}\}$ is expressed as $\hat{r}_d = \bra{e}\hat{r}_d\ket{g} \ket{e}\bra{g} + \bra{g}\hat{r}_d\ket{e} \ket{g}\bra{e}$, where $\bra{e}\hat{r}_d\ket{g} = \bra{g}\hat{r}_d\ket{e}^*$. 
Given this operator, the interaction Hamiltonian $\hat{V} = q\hat{\mb{r}}^T\mb{E}(t)$ is transferred from the interaction energy $V=q\mb{r}^T\mb{E}(t)$ as
\begin{align}\label{eq:InterHamiltonian}
    \hat{V} =
     \sum_{d\in\{x,y,z\}}q\hat{r}_d E_{d}(t) = \left(\begin{array}{cc}
      0   &  \xi(t)  \\
        \xi^*(t)  & 0
    \end{array}\right),
\end{align}
where  $\xi(t) \overset{\Delta}{=}(\boldsymbol{\mu}_{eg}^T\boldsymbol{\epsilon})\sqrt{P}\rho s\cos(\omega t + \phi)$ and $\boldsymbol{\mu}_{eg}  = q\bra{e}\hat{\mb{r}}\ket{g} = q[\bra{e}\hat{r}_x\ket{g}, \bra{e}\hat{r}_y\ket{g}, \bra{e}\hat{r}_z\ket{g}]^T$. 
The vector $\boldsymbol{\mu}_{eg}$ is called the transition dipole moment~\cite{AtomicPhysics}. 
\end{itemize}


Assume the initial state is $\alpha_g(0) = 1$ and $\alpha_e(0) = 0$. 
We can solve the Schrödinger equation in \eqref{eq:Schrodinger} using the method of rotating-wave approximation (RWA)~\cite{QuantumOptics_Fox2006}. As a result, the probability, $|\alpha_e(t)|^2$, is derived as 
\begin{align} \label{eq:RabiOsc}
    |\alpha_e(t)|^2 = \frac{\Omega_R^2}{\Omega_R^2 + \delta^2} \sin^2\left(\frac{\sqrt{\Omega_R^2 + \delta^2}}{2} t\right),
\end{align}
where $\delta \overset{\Delta}{=} \omega_{eg} - \omega$ denotes the \emph{detuning} of the carrier frequency $\omega$ from the transition frequency $\omega_{eg}$. In \eqref{eq:RabiOsc}, $\Omega_R \overset{\Delta}{=} \frac{|(\boldsymbol{\mu}_{eg}^T\boldsymbol{\epsilon}) \sqrt{P}\rho s|}{\hbar}$ denotes  the \emph{Rabi frequency}, defined as the frequency at which the probability amplitudes of two atomic energy levels fluctuate in response to an oscillating electromagnetic field~\cite{QuantumOptics_Fox2006}. 
For ease of notation, we introduce $\Omega \overset{\Delta}{=} \sqrt{\Omega^2_R + \delta^2}$ to represent the \emph{effective Rabi frequency}. 

The oscillation at the Rabi frequency as depicted in \eqref{eq:RabiOsc}, known as \emph{Rabi oscillation}, results from electrons' back-and-forth transitions between ground and excited states. This causes time-variation of electron populations over energy levels. When the carrier frequency exactly matches the transition frequency (i.e., $\omega = \omega_{eg}$), the excited-state probability in \eqref{eq:RabiOsc} is simplified as $|\alpha_e(t)|^2=\sin^2\left(\frac{\Omega_R}{2}t\right)$ with the oscillation magnitude at its maximum of one. On the other hand, as the carrier frequency $\omega$ deviates from $\omega_{eg}$, the magnitude, given as $\frac{\Omega_R^2}{\Omega_R^2 + \delta^2}$, reduces as the detuning factor, $\delta$, increases; the effective Rabi frequency $\Omega$, however, grows. The above facts reflect the ability of a resonant (or near-resonant) electromagnetic wave to stimulate significant changes on the quantum states of Rydberg atoms. 

\subsection{Integrated Demodulation and Down-conversion via Rabi Oscillation}
The aforementioned Rabi oscillation in Rydberg atoms provides a mechanism for realizing integrated demodulation and downconversion, which are the two key operations of the atomic receiver. The details are described as follows.

    \subsubsection{Demodulation} The phenomenon of Rabi oscillation is useful for demodulation because the transmit symbol, $s$, and carrier frequency, $\omega$, are encoded in the Rabi frequency, $\Omega_R$, and effective Rabi frequency, $\Omega$. Consider the realization of amplitude demodulation first of all,  $\Omega_R$ and $\Omega$ are proportional to the norm of $s$. Therefore, the stronger the symbol energy $|s|^2$, the more frequent the quantum jump occurs. 
    This property can be exploited by atomic receivers to infer amplitude-modulated signals by reading the Rabi frequencies. Next, consider frequency demodulation. Though the strength of $s$ is fixed, the carrier frequency, $\omega$, and the detuning factor, $\delta = \omega_{eg} - \omega$, change with time. 
    Note that the changes on $\omega$ according to data signals cause $\delta$ to vary. Thus one can use the atomic receiver to recover data bits from frequency-modulated signals by measuring the variations in $\Omega = \sqrt{\Omega_R^2 + \delta^2}$.

\subsubsection{Simultaneous down-conversion} The Rabi oscillation in \eqref{eq:RabiOsc} can double as a down-converter. Specifically, in the course of quantum jump, the symbol $s$ modulated on the electromagnetic wave is naturally down-converted to the Rabi frequency without any electronic mixer~\cite{RydMag_Liu2023}. 
Its integration with the preceding demodulation process represents another advantage of atomic receiver. 


\ifx\onecol\undefined
\begin{figure}
   \centering
   \includegraphics[width=3.2in]{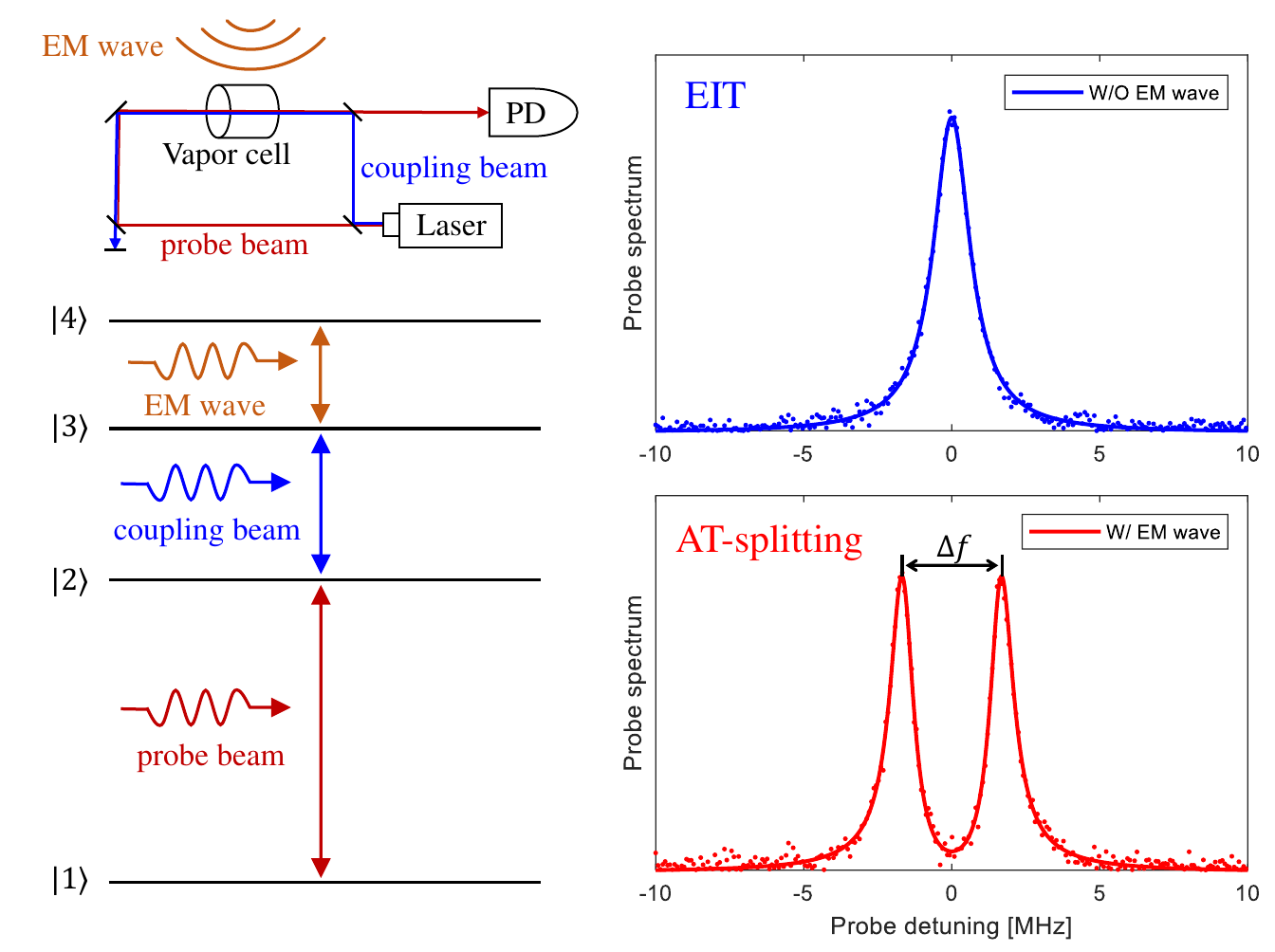}
   \vspace*{-1em}
   \caption{The EIT-AT phenomenon of a four-level system, where $\ket{3}$ and $\ket{4}$ represent the Rydberg states. For example, in literature~\cite{RydAMFM_Anderson2021}, $\ket{1}$, $\ket{2}$, $\ket{3}$, and $\ket{4}$ correspond to the energy levels $6S_{1/2}$, $6P_{3/2}$, $47S_{1/2}$, and $47P_{1/2}$ for measuring $37.4065\:{\rm GHz}$ electromagnetic wave.}
   \label{img:EIT}
   \vspace*{-1em}
\end{figure}
\else
\begin{figure}
   \centering
   \includegraphics[width=4in]{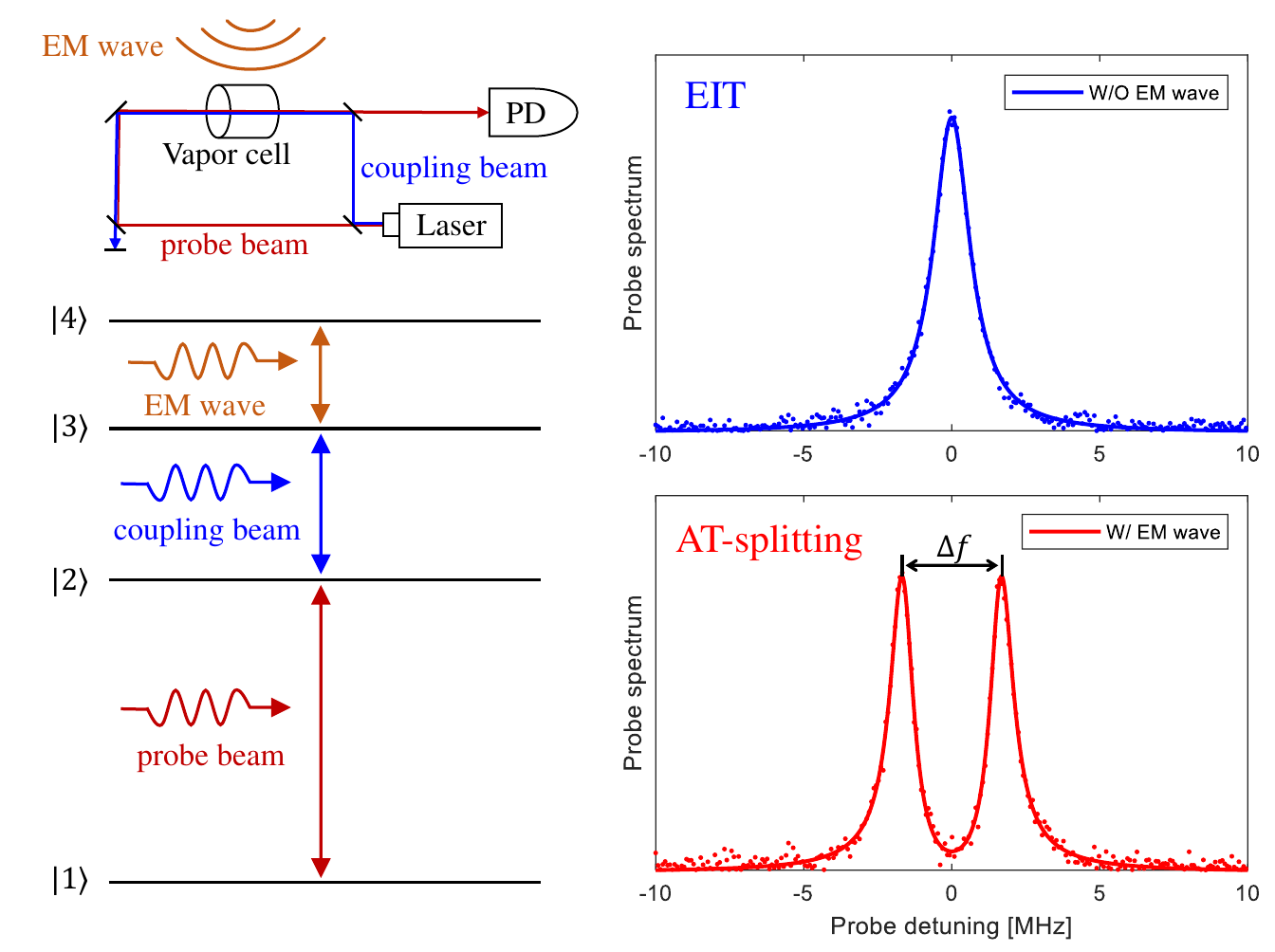}
   \vspace*{-1em}
   \caption{  The EIT-AT phenomenon of a four-level system, where $\ket{3}$ and $\ket{4}$ represent the Rydberg states. For example, in literature~\cite{RydAMFM_Anderson2021}, $\ket{1}$, $\ket{2}$, $\ket{3}$, and $\ket{4}$ correspond to the energy levels $6S_{1/2}$, $6P_{3/2}$, $47S_{1/2}$, and $47P_{1/2}$ for measuring $37.4065\:{\rm GHz}$ electromagnetic wave. 
   		}
   \label{img:EIT}
   \vspace*{-1em}
\end{figure}
\fi

\subsection{Implementation of Atomic Receiver}\label{sec:2D}
Despite the elegant mathematical representation of Rabi oscillation in \eqref{eq:RabiOsc}, observing this phenomenon in practice is quite challenging due to various external disturbances that can destroy the coherence of a quantum state, such as spontaneous emission~\cite{QuantumOptics_Fox2006}. Fortunately, the data-carrying effective Rabi frequency is still measurable, as it can induce two other quantum phenomena, namely \emph{electromagnetically induced transparency} (EIT) and \emph{Autler-Townes} (AT) splitting~\cite{RydMag_Fancher2021, liao_microwave_2020}.
As shown in Fig.~\ref{img:EIT}, the EIT-AT phenomenon involves a four-level system denoted by quantum states $\ket{1}$, $\ket{2}$, $\ket{3}$, and $\ket{4}$, where $\ket{3}$ and $\ket{4}$ represent the Rydberg states we studied before. 
Recall that when preparing a Rydberg state, two laser beams, the probe beam and coupling beam, are needed to excite the electron from $\ket{1}$ to $\ket{2}$ and then from $\ket{2}$ to $\ket{3}$. 
In practice, a PD is deployed to receive the probe beam and illustrate its spectrum for observing the Rabi frequency. 
When the incident electromagnetic wave is turned off and the probe beam is turned resonant with the transition $\ket{1} \rightarrow \ket{2}$, the EIT phenomenon occurs such that the probe beam penetrates through the vapor cell nearly without absorption loss. 
What reflects on its spectrum is that one peak appears at the resonant frequency (see the upper sub-figure of Fig.~\ref{img:EIT}). 
On the other hand, when the incident electromagnetic wave is turned on, the two Rydberg states $\ket{3}$ and $\ket{4}$ are coupled by the Rabi frequency. 
As a result, the AT splitting phenomenon is induced, which splits the peak of the spectrum of the probe beam into two as depicted in the lower sub-figure of Fig.~\ref{img:EIT}. 
As physicists have demonstrated, the splitting interval $\Delta f$ of the two peaks is linearly proportional to the effective Rabi frequency~\cite{RydMag_Liu2023}:
\begin{align}
        \Delta f = \frac{\lambda_c}{\lambda_p} \frac{\Omega}{2\pi},
\end{align}
where $\lambda_c$ and $\lambda_p$ denote the wavelengths of the coupling and probe beams, respectively. Then, measuring $\Delta f$, the effective Rabi frequency, $\Omega = 2\pi\frac{\lambda_p}{\lambda_c
}\Delta f$, can be read out to infer the transmitted symbol. 


\section{Overview of Atomic MIMO Receiver}\label{sec:3}
In the preceding section, we introduced the principles of atomic SISO receiver. In this section, building on these principles, we propose the atomic MIMO receiver for multi-user communications. First, the transmission model of atomic MIMO receivers is presented. Then, the holographic-phase sensing measurement is incorporated into atomic MIMO receivers for detecting \emph{quadrature amplitude modulation} (QAM) symbols, followed by the model of quantum noise.


\subsection{Transmission Model of Atomic MIMO Receiver} \label{sec:3A}

\ifx\onecol\undefined
\begin{figure}
   \centering
   \includegraphics[width=3.2in]{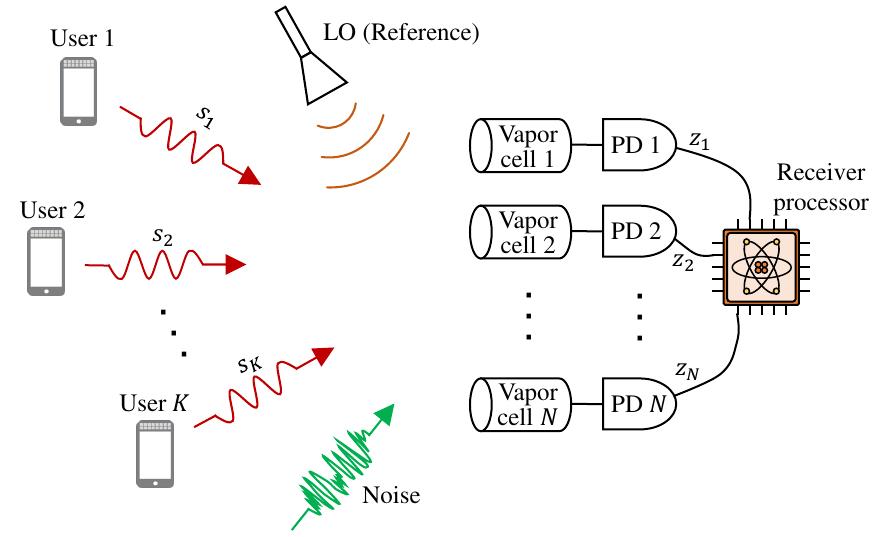}
   \caption{Schematic diagram of the atomic MIMO receiver. One atomic receiver employing $N$ atomic antennas is measuring signals from $K$-single antenna users together with one reference source.}\label{img:MIMO}
   \vspace*{-1em}
\end{figure}
\else
\begin{figure}
   \centering
   \includegraphics[width=3.5in]{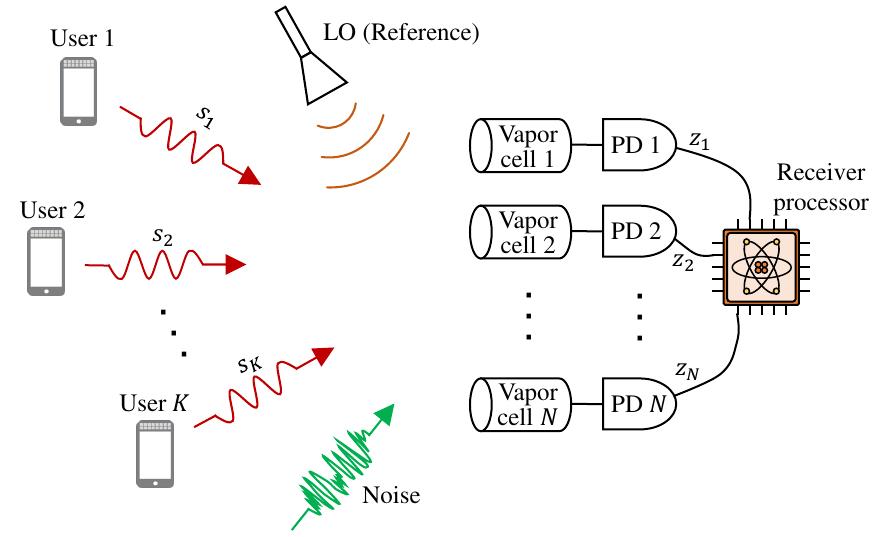}
   \vspace*{-1em}
   \caption{Schematic diagram of the atomic MIMO receiver. One atomic receiver employing $N$ atomic antennas is measuring signals from $K$-single antenna users together with one reference source.}\label{img:MIMO}
   \vspace*{-1em}
\end{figure}
\fi

Consider the multi-user communication system illustrated in Fig.~\ref{img:MIMO}, where $K$ single-antenna users communicate with one atomic MIMO receiver. Its architecture comprises $N$ vapor cells filled with Rydberg atoms and an equal number of PDs.  
Each vapor cell acts as an \emph{atomic antenna} for measuring electromagnetic waves, while the PDs read out the measurement results (i.e., the splitting interval, $\Delta_f$, and the Rabi frequency, $\Omega_R$) and aggregate them to the digital processor for signal detection. 
We assume all users share the same carrier frequency $\omega$, while all atoms are triggered to the uniform Rydberg states $\hbar \omega_g$ and $\hbar \omega_e$ (e.g., $52D_{5/2}$ and $53 P_{3/2}$) in response to the carrier frequency (e.g., $\omega = 5$~GHz). As a result, the resonance case with $\omega = \omega_{eg}$ is created. 

Consider an arbitrary user, say the $k$-th user. Its phase modulated symbol is denoted as $s_k = |s_k|e^{i\gamma_k} \in \mathcal{S}$, where the set $\mathcal{S}$ corresponds to the QAM constellation. 
The average power of $s_k$ is normalized as $\mathsf{E}(\|s_k\|^2) = 1$. 
Then, the transmit RF signal is given by $x_k(t) = \sqrt{P_k}|s_k|\cos(\omega t + \gamma_k)$. 
{\color{black} We adopt the practical multi-path wireless channel with $L$ 
paths.  The simultaneous transmissions from all users result in the received 
electromagnetic wave at the $n$-th vapor cell:  
\begin{align}\label{eq:Ent}
    \mb{E}_n(t) = \sum_{k = 1}^K\sum_{l = 1}^L \boldsymbol{\epsilon}_{nkl} \sqrt{P_k} \rho_{nkl} |s_k| \cos (\omega t + \phi_{nkl} + \gamma_k),
\end{align}
where variables $\boldsymbol{\epsilon}_{nkl}$, $\rho_{nkl}$, and $\phi_{nkl}$ represent the polarization, path loss, and phase shift of the propagation from the $k$-th user to the $n$-th antenna via the $l$-th path.}

Next, consider an arbitrary vapor cell, say the $n$-th cell at the receiver. For each atom inside this cell,  its quantum state can be expressed as $\ket{\psi(t)}_n = \alpha_{g,n}(t) \ket{g} + \alpha_{e,n}(t) \ket{e} = [\alpha_{e,n}(t), \alpha_{g,n}(t)]^T$. Similarly as in the SISO case, the evolution of $\ket{\psi(t)}_n$ is governed by the Schrödinger equation
\begin{align}\label{eq:scheqn}
    i\hbar \frac{\partial \ket{\psi(t)}_n}{\partial t} = (\hat{H}_n + \hat{V}_n)\ket{\psi(t)}_n.
\end{align} 
As all atoms are exited to the uniform Rydberg states, the free Hamiltonian $\hat{H}_n $ is the same as \eqref{eq:FreeHamiltonian}, i.e., $\hat{H}_n = \diag\{\hbar\omega_e, \hbar\omega_g\}$. 
Similar to the discussion in \eqref{eq:InterHamiltonian}, the interaction Hamiltonian $\hat{V}_n$ results from the coupling of the incident wave $\mb{E}_n(t)$ and the position operator $\hat{\mb{r}}_n$, i.e., $\hat{V}_n = q\hat{\mb{r}}_n^T\mb{E}_n(t) = \sum_{k = 1}^K\sum_{l = 1}^L(q\hat{\mb{r}}_n^T\boldsymbol{\epsilon}_{nkl}) \sqrt{P_k}\rho_{nkl} |s_k| \cos (\omega t + \phi_{nkl} + \gamma_k)$.
By expanding the position operator $\hat{\mb{r}}_n$ on the basis $\{\ket{g}, \ket{e}\}$ as $\hat{\mb{r}}_n = \bra{e}\hat{\mb{r}}_n\ket{g}\ket{e}\bra{g} + \bra{g}\hat{\mb{r}}_n\ket{e}\ket{g}\bra{e}$, we arrive at the matrix-form expression: 
\begin{align}
\hat{V}_n = \xi_n(t)\ket{e} \bra{g} + \xi_n^*(t)\ket{g} \bra{e},
\end{align}
where 
\begin{align}\label{eq:xin}\color{black}
\xi_{n}(t) {=} \sum_{k = 1}^K\sum_{l = 1}^L{\boldsymbol{\mu}}_{eg}^T\boldsymbol{\epsilon}_{nkl}  \sqrt{P_k}\rho_{nkl} |s_k| \cos (\omega t + \phi_{nkl} + \gamma_k). 
\end{align}
We assume that the initial state is $\alpha_{e,n} = 0$. Then, we can derive the
probability amplitude using RWA (see details in Appendix A) as
\begin{align} \label{eq:RabiMIMO}
    |\alpha_{e,n}(t)|^2 
    = \sin^2\left(\frac{\Omega_{R,n}}{2} t\right),
\end{align}
where the Rabi frequency is given as \begin{align}\label{eq:Sch2}\color{black}
    \Omega_{R,n} = \left|\sum_{k=1}^{K} \sum_{l = 1}^L\frac{1}{\hbar} ({\boldsymbol{\mu}}_{eg}^H\boldsymbol{\epsilon}_{nkl}) \sqrt{P_k}\rho_{nkl} |s_k| e^{i(\phi_{nkl} + \gamma_k)}\right|.
\end{align} 
Note that the detuning factor, $\delta$, vanishes since $\delta = \omega_{eg} - \omega = 0$. The preceding formula characterizes the quantum-jump behavior of the $n$-th atomic antenna as triggered by $K$ sources, where multi-user  symbols, $\{s_k\}$, 
are encoded in the Rabi frequencies, $\{\Omega_{R,n}\}$. 

The goal of atomic-MIMO receiver is to read the values of $\{\Omega_{R,n}\}$ and infer $\{s_k\}$ from $\{\Omega_{R,n}\}$.
To this end, according to the principle of EIT-AT discussed in Section \ref{sec:2D}, the Rabi frequency, $\Omega_{R,n}$, can be measured by the splitting interval, $\Delta f_n$, of the probe-beam spectrum since $
    \Omega_{R,n} = 2\pi\frac{\lambda_p}{\lambda_c}\Delta f_n$. 
   For ease of exposition, let the measurement of $\Omega_{R,n}$ be denoted as $z_n$, which is the equivalent received signal at the $n$-th antenna.  
To facilitate the detector design, a compact matrix form of $\{z_n\}$ is 
derived as follows. We define $\mb{s} = [s_1, s_2, \cdots, s_K]^T$ and 
$\mb{a}_n = [a_{n,1}, a_{n,2}, \cdots, a_{n,K}]^T$, {\color{black} where 
$a_{n,k} = \sum_{l = 1}^L\frac{1}{\hbar} 
({\boldsymbol{\mu}}_{eg}^T\boldsymbol{\epsilon}_{n,k})  \sqrt{P_k}\rho_{nkl} 
e^{-i\phi_{nkl}}$}. It is straightforward to obtain that $z_{n} = \left| 
\mb{a}_n^H \mb{s} \right|$.
We further define the received vector-form signal as $\mb{z} = [z_1, z_2, \cdots, z_N]^T$ and the observation matrix as $\mb{A} = [\mb{a}_1, \mb{a}_2, \cdots, \mb{a}_N] \in \mathbb{C}^{K \times N}$. It follows that 
\begin{align}\label{eq:z=|As|}
    \mb{z} = |\mb{A}^H\mb{s}|,
\end{align}
where the operator $|\mb{x}|$ takes the magnitude of each element of the vector, $\mb{x}$.
A crucial conclusion is drawn from \eqref{eq:z=|As|} that the transmit symbols in $\mb{s}$ are embedded in the observed amplitude vector $\mb{z} = |\mb{A}^H\mb{s}|$, where the observation matrix $\mb{A}$ can be treated as an \emph{effective MIMO channel}. 
As opposed to the linear model adopted in conventional MIMO communications, the atomic MIMO receiver features a non-linear model based on amplitude detection. This is equivalent to the well-known PR model in the realm of optical imaging~\cite{PR_Dong2023}. 
Before delving into this PR problem, two factors, the holographic-phase sensing measurement and the noise model, that are part of the proposed framework and affect the receiver's effectiveness should be discussed.

\subsection{Holographic-Phase Sensing Measurement}
One can observe from \eqref{eq:z=|As|} that the absolute phase $e^{i\angle\mb{s}}$ of $\mb{s}$ is ambiguous to the received signal $\mb{z}$, making it hard to detect QAM symbols. To address this issue, a measurement scheme called \emph{holographic phase-sensing} is often employed~\cite{RydPhase_And2020}. As illustrated in Fig.~\ref{img:MIMO}, this scheme artificially introduces a LO to send known reference signals, denoted as $s_b = |s_b|e^{j\gamma_b}$, to interfere with signals from users. The LO is typically implemented by a horn antenna and can be deployed close to the receiver. 
Utilizing the reference signal $s_b$, one can infer the absolute phase, $e^{i\angle\mb{s}}$, by estimating the phase difference between $\mb{s}$ and $s_b$, thereby enabling QAM detection.  

Mathematically speaking, let the reference signal share the same carrier frequency $\omega$ with the user signal.
The transmit RF signal by LO is expressed as $x_b(t) = \sqrt{P_b}|s_b|\cos(\omega t + \gamma_b)$. Since the LO-to-receiver distance is quite short, e.g., tens of centimeters, the LO-to-receiver channel is dominated by the line-of-sight (LoS) path. Thereby, the received electromagnetic wave in \eqref{eq:Ent} should be modified as 
\begin{align}\label{eq:Ent2}
    \mb{E}_n(t) = \sum_{k = 1}^K\sum_{l = 1}^L \boldsymbol{\epsilon}_{nkl} \sqrt{P_k} \rho_{nkl} |s_k| \cos (\omega t + \phi_{nkl} + \gamma_k) \notag \\
    + \boldsymbol{\epsilon}_{b,n} \sqrt{P_b} \rho_{b,n} |s_b| \cos (\omega t + \phi_{b,n} + \gamma_b),
\end{align}
where $\rho_{b,n}$, $\phi_{b,n}$, and $\boldsymbol{\epsilon}_{b,n}$ refer to the path loss, phase shift, and polarization of the LO-to-receiver LoS path, respectively. Then, following the same steps as \eqref{eq:scheqn} to \eqref{eq:z=|As|}, we can arrive at the modified received signal: 
\begin{align}\label{eq:BPR}
    \mb{z} = |\mb{A}^H\mb{s} + \mb{b}|,
\end{align}
where $\mb{b}=[b_1,\cdots,b_N]^T$ and $b_{n} = \frac{s_b}{\hbar} {\boldsymbol{\mu}}_{eg}^H\boldsymbol{\epsilon}_{b,n} \sqrt{P_b}\rho_{b,n} e^{i\phi_{b,n}}$. 
Due to the presence of reference signal $\mb{b}$, \eqref{eq:BPR} is called a \emph{biased PR model} as opposed to the typical PR model without this reference signal. 

{\color{black}
\subsection{Noise Model}
The noise perturbing the atomic MIMO receiver is attributed to both external and internal sources. 

External noise is mainly induced by the signal interference from adjacent time slots, bands, and nearby cells, which is clearly environment-dependent. We use the notation $\sigma_E^2$ to represent the total variance of external noise. 

The internal noise comes from the thermal noise of electronic components and the randomness of quantum measurement, i.e., \emph{quantum shot noise} (QSN). 
In contrast to conventional RF systems, an atomic receiver makes no use of RF circuits at the front end that produce thermal noises, such as analog filters, amplifiers, and mixers. This fact means that an atomic receiver is immune to thermal noise. Thereafter, the limiting noise source for atomic receivers is the QSN, whose level in the unit of electric field strength [V/m] is presented by \cite{RydMag_Fancher2021}
\begin{align}
    E_{\min} = \frac{h}{\|\boldsymbol{\mu}_{eg}\|(T_iT_rN_a)^{1/2}},
\end{align}
where $T_i$, $T_r$, and $N_a$ denote the measurement time, the relaxation time of Rydberg states, and the number of participating atoms, respectively. 
By plugging this noise limit into the electric field and further transforming it into the unit of Rabi frequency [Rad/s],
$E_{\min}$ is normalized to $\frac{\|\boldsymbol{\mu}_{eg}\|} {\hbar} E_{\min} = \frac{2\pi}{(T_tT_rN_a)^{1/2}}$. As a result, the variance of QSN is given as
\begin{align}\label{eq:sigmas}
    \sigma_S^2 = \frac{4\pi^2}{T_iT_rN_a}. 
\end{align}
The extreme sensitivity of atomic receivers arises because the noise level \eqref{eq:sigmas} can be orders of magnitude lower than the thermal noise of conventional receivers \cite{QuanSense_Zhang2023}. 
To fairly compare the quantum shot noise and thermal noise, we adopt the standard $T_i = 1\:{\rm s}$ measurement time for both them. In this case, the power density of room-temperature thermal noise is $-176\:{\rm dBm}$~\cite{QSN_Bussey2022}. As for the quantum noise, we adopt the typical parameters $T_r = 10\:\mu{\rm s}$ and $N_a = 5\times 10^5$, which give rise to a noise level $\frac{2\pi}{(T_iT_rN_a)^{1/2}} = 2.81\:{\rm Hz}$. According to transformation between electric field and power~\cite{QSN_Bussey2022}, this noise refers to a power density of $-191\:{\rm dBm}$ when $\omega = 2\pi\times 5\:{\rm GHz}$, which is 15 dB smaller than the room-temperature thermal noise level. This reflects the high sensitivity of atomic receivers~\cite{QuanSense_Zhang2023}.

In summary, the total variance of noise is expressed as 
\begin{align}
    \sigma^2 = \sigma_S^2 + \sigma_E^2 = \frac{4\pi^2}{T_iT_rN_a} + \sigma_E^2. 
\end{align}
We invoke the law-of-large-number to model the noise as an additive Gaussian random vector, $\mb{w} \sim \mathcal{CN}(0, \sigma^2 \mb{I})$.
By adding the noise into \eqref{eq:Ent2}, we finally arrive at the following input-output relationship for atomic MIMO systems: 
\begin{align}\label{eq:BiasedPR}
    \mb{z} = |\mb{A}^H\mb{s} + \mb{b} + \mb{w}|.
\end{align}
}

\section{Atomic-MIMO Signal Detection}\label{sec:4}
In this section, we first formulate the signal detection problems of atomic MIMO receivers based on the LS and ML principles. Then, two atomic-MIMO detection algorithms are proposed to solve the formulated problems. Their performance and complexity are subsequently analyzed.

\subsection{Atomic-MIMO Signal Detection Problems}\label{sec:3B}
The objective of atomic-MIMO signal detection is to estimate the multi-user data vector, $\mb{s}$, from the received signal, $\mb{z}$, in \eqref{eq:BiasedPR}.
Two problem formulations are provided below.
\subsubsection{Least-Square Criterion} The first formulation follows the typical LS criterion:
\begin{align}\label{eq:LS}
    \min_{\mb{s} \in \mathcal{S}^K} \left\| \mb{z} - |\mb{A}^H\mb{s} + \mb{b}| \right\|_2^2.
\end{align}

\subsubsection{Maximum-Likelihood Criterion} 
It can be proven that each entry of $\mb{z}$ follows the Rician distribution, i.e.,
\begin{align} \label{eq:Rician}
    p(z_n; \mb{s}) = \frac{2z_n}{\sigma^2}\exp\left(- \frac{z_n^2 + |\lambda_n|^2}{\sigma^2} \right) I_0 \left(\kappa_n\right), 
\end{align}
where $\lambda_n \overset{\Delta}{=} \mb{a}_n^H\mb{s} + b_n$, $\kappa_n \overset{\Delta}{=}\frac{2z_n|\mb{a}_n^H\mb{s} + b_n|}{\sigma^2}$, and
$I_0(x) = \frac{1}{2\pi} \int_0^{2\pi} e^{x\cos\theta}\text{d}\theta$ is the modified zero-order Bessel function. 
Thereby, for any given realization of the effective channel and reference, the ML detection of $\mb{s}$
is formulated as 
\begin{align}\label{eq:ML}
    \max_{\mb{s} \in \mathcal{S}^K} \sum_{n = 1}^N \log p(z_n; \mb{s}) \propto \sum_{n = 1}^N  - \frac{|\lambda_n|^2}{\sigma^2} + \log I_0 \left(\kappa_n \right).
\end{align}

Although problems \eqref{eq:LS} and \eqref{eq:ML} can be optimally solved by an exhaustive search, the complexity is prohibitive when the users are many and/or the constellation size is large. To this end, two low-complexity algorithms will be designed subsequently.
\begin{algorithm}[tb]
	\caption{$\!\!$: Biased GS Algorithm}
 \label{alg1}
	\begin{algorithmic}[1]
		\REQUIRE ~ 
		The equivalent received signal $\mb{z}$, 
  the equivalent channel matrix $\mb{A}$, the reference field $\mb{b}$, the total number of iterations $t_0$.
		\STATE Construct the augmented observation matrix $\bar{\mb{A}} = [\mb{A}^H, \mb{b}]^H = [\bar{\mb{a}}_1,\bar{\mb{a}}_2,\cdots \bar{\mb{a}}_N]$ \\
  \STATE Find the principal  eigenvector $\mb{v}$ of $\mb{M} = \sum_{n=1}^N z_n \bar{\mb{a}}_n\bar{\mb{a}}_n^H$
  \STATE Set $\bar{r} = \frac{ |\mb{v}^H\bar{\mb{A}}| \mb{z}}{\|\bar{\mb{A}}^H\mb{v}\|_2^2}$ and $\bar{\mb{s}}^0 = \bar{r} \mb{v}$ 
  \\
  \STATE Initialize $\mb{s}^0 =  e^{-i\angle\bar{\mb{s}}^0(K + 1)}\bar{\mb{s}}^0(1:K) $
  \\
  \FOR{$t = \{1,2,\cdots,t_0\}$}
    \STATE Update $\boldsymbol{\theta}^t = \angle(\mb{A}^H\mb{s}^{t-1} + \mb{b})$
    \STATE Update $\mb{s}^t = (\mb{A}\mb{A}^H)^{-1}\mb{A}(\mb{z}\circ e^{i\boldsymbol{\theta}^t} - \mb{b})$
  \ENDFOR
  \ENSURE ~ The estimated symbol $\tilde{\mb{s}} = \mb{s}^{t_0}$\\
	\end{algorithmic}
\end{algorithm}

\subsection{Atomic-MIMO Signal Detection Based on LS Criterion}\label{sec:4A}
Consider the biased PR problem in \eqref{eq:LS} corresponding to LS detection. 
One of the most popular PR solvers is the Gerchberg-Saxton (GS) algorithm~\cite{PR_Netrapalli2015}. We transform the classical GS algorithm to a \emph{biased GS algorithm} to solve the biased PR problem in \eqref{eq:LS} . The result is summarized in Algorithm \ref{alg1} and elaborated as follows.

Define $\mb{y} = [y_1, y_2,\cdots,y_N]^T = \mb{A}^H\mb{s} + \mb{b} + \mb{w}$. The fundamental idea of the GS algorithm is that if the true phase $e^{i\boldsymbol{\theta}}$ of $\mb{y}$ is accessible, where  $e^{i\boldsymbol{\theta}}= [e^{i\angle{y_1}},e^{i\angle {y_2}},\cdots,e^{i\angle {y_N}}]^T$, then problem \eqref{eq:LS} becomes a linear regression problem: $\min_{\mb{s}}\|\mb{z}\circ e^{i\boldsymbol{\theta}} - \mb{A}^H\mb{s} -  \mb{b}\|$, whose solution is $\mb{s} = (\mb{A}\mb{A}^H)^{-1}\mb{A}(\mb{z}\circ e^{i\boldsymbol{\theta}} - \mb{b})$, where $\circ$ denotes the Hadamard product.  As the knowledge on $e^{i\boldsymbol{\theta}}$ is unavailable, an alternative approach is to jointly optimize $e^{i\boldsymbol{\theta}}$ and $\mb{s}$:
\begin{align}\label{eq:GS}
    \min_{\boldsymbol{\theta}, \mb{s} \in \mathcal{S}^K} \left\| \mb{z}\circ e^{i\boldsymbol{\theta}} - \mb{A}^H\mb{s} - \mb{b} \right\|_2^2.
\end{align}

Problem \eqref{eq:GS} is non-convex since $e^{i\boldsymbol{\theta}}$ is an oscillating function with respect to (w.r.t) $\boldsymbol{\theta}$. Fortunately, this optimization problem can be efficiently solved by the well-known alternating minimization methodology. To be specific, $\boldsymbol{\theta}$ and $\mb{s}$ are iteratively updated as presented in steps 6-7 of Algorithm \ref{alg1}. 
During the $t$-th iteration, given the previous estimate $\mb{s}^{t-1}$, $\boldsymbol{\theta}$ can be  optimized by extracting the phase of $\mb{A}^H\mb{s}^{t-1} + \mb{b}$, that is $
    \boldsymbol{\theta}^t = \angle(\mb{A}^H\mb{s}^{t-1} + \mb{b})$.
Furthermore, given $\boldsymbol{\theta}^{t}$, the optimal $\mb{s}$ is determined by the LS rule:
$
    \mb{s}^t = (\mb{A}\mb{A}^H)^{-1}\mb{A}(\mb{z}\circ e^{i\boldsymbol{\theta}^{t}} - \mb{b})
$.

One issue in the implementation of biased GS algorithm is its sensitivity to the choice of initial point. Without a proper initial estimate, Algorithm \ref{alg1} may converge to a shallow local optimum. To address this issue, the computationally efficient spectral method has been developed to yield good initialization~\cite{PR_Candes2015}. Specifically, to adapt the spectral method into our model, we first construct an augmented observation matrix $\bar{\mb{A}} = [\mb{A}^H, \mb{b}]^H $ as presented in step 1. This matrix allows us to treat the user signal and reference signal together: $\bar{\mb{s}} = [\mb{s}^H, 1]^H$. Thereafter the observed amplitude vector can be rewritten in a reference-free form: $\mb{z} = |\mb{A}^H\mb{s} + \mb{b} + \mb{w}| = |\bar{\mb{A}}^H\bar{\mb{s}} +  \mb{w}|$, whose initialization can be performed using the typical spectral method. To elaborate, define the weighted covariance matrix as $\mb{M} = \sum_{n=1}^N z_n \bar{\mb{a}}_n\bar{\mb{a}}_n^H$, where $\bar{\mb{a}}_n$ denotes the $n$-th column of $\bar{\mb{A}}$. The direction of the initial value $\bar{\mb{s}}^0$ of $\bar{\mb{s}}$ is estimated as the principal eigenvector $\mb{v}$ of $\mb{M}$. In addition, the magnitude $\bar{r}$ of $\bar{\mb{s}}^0$ is obtained by approaching $|\mb{A}^H\bar{\mb{s}}^0| = |\mb{A}^H\bar{r} \mb{v}|$ to $\mb{z}$:
\begin{align}
    \bar{r} = \arg\min_{r'} \| {r'}|\mb{A}^H\mb{v}| - \mb{z}\|_2^2 =  \frac{ |\mb{v}^H\bar{\mb{A}}| \mb{z}}{\|\bar{\mb{A}}^H\mb{v}\|_2^2}.
\end{align}
Thereafter, $\bar{\mb{s}}$ is initialized as $\bar{\mb{s}}^0 = \bar{r}\mb{v}$. Recall that only the first $K$ entries of $\bar{\mb{s}}$ are the wanted symbols from users, while the last entry of $\bar{\mb{s}}$ refers to the known value 1 with a zero phase. Hence,  $\mb{s}^0$ is initialized by forcing the phase of the last entry of $\bar{\mb{s}}^0$ to zero:
\begin{align}
    \mb{s}^0 =  e^{-i\angle\bar{\mb{s}}^0(K + 1)}\bar{\mb{s}}^0(1:K),
\end{align}
which completes the initialization. After carrying out the initialization and $t_0$-step iterations, the final estimate is obtained as $\tilde{\mb{s}} = \mb{s}^{t_0}$. Last, one can project each entry of $\tilde{\mb{s}}$ to the nearest point of the QAM constellation to perform demapping, i.e., $\hat{s}_k = \arg\min_{s\in \mathcal{S}} \|\tilde{s}_k - s\|$.

The main limitation of the LS based biased GS algorithm arises as it fails to consider the distribution of the received signal, leaving margin for performance improvement. The limitation can be overcome by the ML algorithm in the sequel. 


\begin{algorithm}[tb]
	\caption{$\!\!$:  EM-GS Algorithm}
	\label{alg2}
	\begin{algorithmic}[1]
		\REQUIRE ~ 
		The equivalent received signal $\mb{z}$, the equivalent channel matrix $\mb{A}$, the reference field $\mb{b}$, the power density of noise $\sigma$, the number of iterations $t_0$.
		\STATE Construct the augmented observation matrix $\bar{\mb{A}} = [\mb{A}^H, \mb{b}]^H = [\bar{\mb{a}}_1,\bar{\mb{a}}_2,\cdots \bar{\mb{a}}_N]$ \\
  \STATE Find the principal  eigenvector $\mb{v}$ of $\mb{M} = \sum_{n=1}^N z_n \bar{\mb{a}}_n\bar{\mb{a}}_n^H$
  \STATE Set $\bar{r} = \frac{ |\mb{v}^H\bar{\mb{A}}| \mb{z}}{\|\bar{\mb{A}}^H\mb{v}\|_2^2}$ and $\bar{\mb{s}}^0 = \bar{r} \mb{v}$ 
  \\
  \STATE Initialize $\mb{s}^0 =  e^{-i\angle\bar{\mb{s}}^0(K + 1)} \bar{\mb{s}}^0(1:K) $
  \\
  \FOR{$t = \{1,2,\cdots,t_0\}$}
    \STATE Update $\boldsymbol{\theta}^t = \angle(\mb{A}^H\mb{s}^{t-1} + \mb{b})$
    \STATE Update $\boldsymbol{\kappa}^t = \frac{2}{\sigma^2}\mb{z}\circ|\mb{A}^H\mb{s}^{t - 1} + \mb{b}|$
    \STATE Update $\mb{s}^t = (\mb{A}\mb{A}^H)^{-1}\mb{A}(\mb{z}\circ e^{i\boldsymbol{\theta}^t} \circ R(\boldsymbol{\kappa}^t) - \mb{b})$
  \ENDFOR
  \ENSURE ~ The estimate symbol $\tilde{\mb{s}} = \mb{s}^{t_0}$\\
	\end{algorithmic}
\end{algorithm}

\subsection{Atomic-MIMO Signal Detection Based on ML Criterion}\label{sec:4B}
Consider the ML detection problem in \eqref{eq:ML}. We propose to solve it by integrating the Expectation-Maximization (EM) and aforementioned GS algorithm, termed the EM-GS algorithm. 
One can observe from \eqref{eq:ML} that the occurrence of the non-trivial 
modified Bessel function $I_0(\cdot)$ makes the ML detection intractable. To overcome this difficulty, the EM method is employed to create a series of tractable lower-bound surrogate functions to the primal objective in \eqref{eq:ML}, resulting in sequentially updated estimates of $\mb{s}^t$ for $t = 1,2,\cdots, t_0$  as summarized in Algorithm \ref{alg2}. The detailed EM-GS algorithm is designed as follows.

To begin with, recall the definition $\mb{y} = \mb{A}^H\mb{s} + \mb{b} + \mb{w} = \mb{z} \circ e^{i\boldsymbol{\theta}}$. 
Like the basic idea of the biased GS algorithm, the phase $e^{i\boldsymbol{\theta}}$ is an unobserved but  crucial variable, the information of which can substantially simplify our problem.  This observation motivates us to treat the phase $e^{i\boldsymbol{\theta}}$ as a latent variable in our EM algorithm. 
Based on the Jensen inequality, the EM algorithm states that the loglikelihood function is lower-bounded by the following surrogate function:
\ifx\onecol\undefined
\begin{align}\label{eq:surrogate}
    \sum_{n = 1}^N \log p(z_n; \mb{s}) &\ge \sum_{n=1}^N \mathsf{E}_{p(\theta_n|z_n, \mb{s}^{t - 1})}\log p(z_n, \theta_n; \mb{s})\notag\\
    &\overset{\Delta}{=} Q(\mb{s}|\mb{s}^{t-1}), 
\end{align}
\else 
\begin{align}\label{eq:surrogate}
    \sum_{n = 1}^N \log p(z_n; \mb{s}) &\ge \sum_{n=1}^N \mathsf{E}_{p(\theta_n|z_n, \mb{s}^{t - 1})}\log p(z_n, \theta_n; \mb{s})\overset{\Delta}{=} Q(\mb{s}|\mb{s}^{t-1}), 
\end{align}
\fi
where $\mb{s}^{t-1}$ denotes the previous estimate of $\mb{s}$. The EM algorithm proceeds iteratively by two steps to maximize the loglikelihood. They are the Expectation step (E-step) to derive the surrogate function $Q(\mb{s}|\mb{s}^{t-1})$ by calculating the expectation in \eqref{eq:surrogate}, and the Maximization step (M-step) to update $\mb{s}^{t}$ by maximizing $Q(\mb{s}|\mb{s}^{t-1})$, i.e., $\mb{s}^{t} = \arg\max_{\mb{s}}Q(\mb{s}|\mb{s}^{t-1})$. These two steps are elaborated below.

\begin{itemize}
\item \textbf{E-step}:
In \eqref{eq:surrogate}, the likelihood function $p(z_n, \theta_n; \mb{s})$ stands for the ML estimator to $\mb{s}$ when both the amplitude $z_n$ and phase $\theta_n$ are accessible. Clearly, $p(z_n, \theta_n; \mb{s})$ is a Gaussian distribution 
\begin{align}\label{eq:Gaussian}
    p(z_n, \theta_n; \mb{s}) &= \frac{1}{\pi\sigma^2} \exp \left(-\frac{1}{\sigma^2}\left|z_ne^{i\theta_n}- \lambda_n\right| \right),
\end{align}
where $\lambda_n \overset{\Delta}{=} \mb{a}_n^H\mb{s} + b_n$.
As for the posterior probability $p(\theta_n|z_n, \mb{s}^{t-1})$, it refers to our current knowledge to the phase $\theta_n$ given previous estimate $\mb{s}^{t-1}$. As shown in~\cite{Mises_Gattoa2007}, $p(\theta_n|z_n, \mb{s}^{t-1})$ follows the von Mises distribution:
\begin{align}\label{eq:von Mises}
    p(\theta_n|z_n, \mb{s}^{t-1}) = \frac{\exp\left(\kappa_n^t \cos(\theta_n - \theta_n^t)\right)}{2\pi I_0(\kappa_n^t)},
\end{align}
where $\theta_n^t \overset{\Delta}{=} \angle\left(\mb{a}_n^H\mb{s}^{t-1} + b_n\right)$ and $\kappa_n^t \overset{\Delta}{=} \frac{2z_n|\mb{a}_n^H\mb{s}^{t-1} + b_n|}{\sigma^2}$. Based on \eqref{eq:Gaussian} and \eqref{eq:von Mises}, we derive the following result that defines the E-step. 
\begin{lemma}\label{lemma1}
\emph{The surrogate function is expressed as
  \begin{align}\label{eq:LR}
     Q(\mb{s}|\mb{s}^{t-1}) =   - \frac{1}{\sigma^2}\| \mb{z}\circ e^{i\boldsymbol{\theta}^t} \circ R(\boldsymbol{\kappa}^t) - \mb{A}^H\mb{s} - \mb{b} \|_2^2 + C,
  \end{align}
  where   $e^{i\boldsymbol{\theta}^t} = [e^{i\theta_1^t}, e^{i\theta_2^t},\cdots, e^{i\theta_N^t}]^T = e^{i\angle(\mb{A}^H\mb{s}^{t-1} + \mb{b})}$, $\boldsymbol{\kappa}^t = [\kappa_1^t, \kappa_2^t,\cdots, \kappa_N^t]^T=\frac{2}{\sigma^2}\mb{z}\circ|\mb{A}^H\mb{s}^{t - 1} + \mb{b}|$, $R(x) \overset{\Delta}{=} \frac{I_1(x)}{I_0(x)}$, $I_1(x)$ represents the first-order modified Bessel function, and $C$ denotes a constant irrelevant to $\mb{s}$.}
\end{lemma}
\begin{proof}
    (See Appendix B). 
\end{proof}


\item \textbf{M-step}:
 Lemma \ref{lemma1} reveals that the maximization of $Q(\mb{s}|\mb{s}^{t-1})$ is equivalent to the minimization of $\| \mb{z}\circ e^{i\boldsymbol{\theta}^t} \circ R(\boldsymbol{\kappa}^t) - \mb{A}^H\mb{s} - \mb{b} \|_2^2$. In this context, the primal non-convex optimization problem in \eqref{eq:ML} is converted to a sequence of easily handled linear regression problems. Therefore, as presented in steps 6-8 of Algorithm \ref{alg2}, the estimate of $\mb{s}$ is updated successively by 
\begin{align}
    \mb{s}^t = (\mb{A}\mb{A}^H)^{-1}\mb{A}(\mb{z}\circ e^{i\boldsymbol{\theta}^t} \circ R(\boldsymbol{\kappa}^t) - \mb{b}), 
\end{align}
where the final estimate is given as $\tilde{\mb{s}} = \mb{s}^{t_0}$. 

\end{itemize}

Moreover, to avoid convergence to a shallow local optimum, the spectral method based on the augmented matrix $\bar{\mb{A}}$ as used in Section~\ref{sec:4A} is carried out to obtain an initial estimate $\mb{s}^{0}$. 
In the end, after performing the initialization and EM-based iteration, the projection operation $\hat{s}_k = \arg\min_{s\in \mathcal{S}} \|\tilde{s}_k - s\|$ is applied to complete the constellation demapping.

\subsection{Comparison between biased GS and EM-GS algorithms}\label{sec:4C}

\ifx\onecol\undefined
\begin{figure}[t!]
\centering
\includegraphics[width=3in]{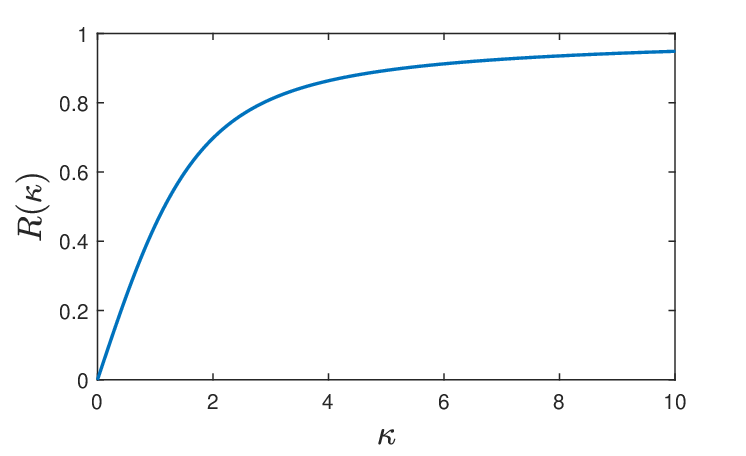}
\vspace*{-1em}
\caption{The high-pass filter corresponding to $R(\kappa)$ for $\kappa\in[0, 10]$.
}
\label{img:RFunction}
\vspace*{-1em}
\end{figure}
\else
\begin{figure}[t!]
\centering
\includegraphics[width=3.5in]{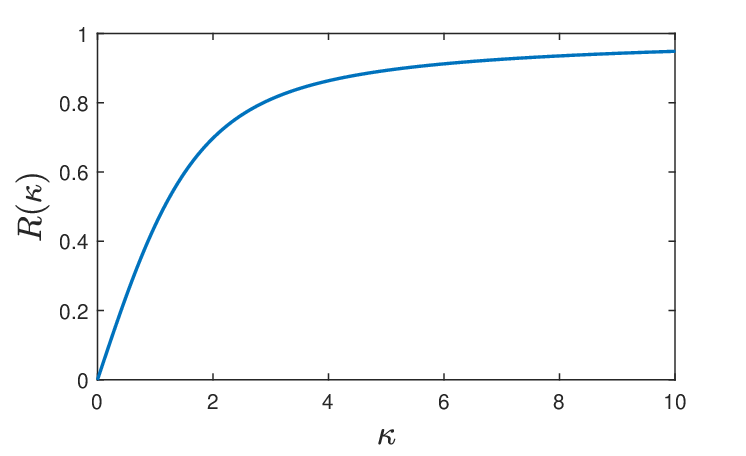}
\vspace*{-1em}
\caption{The high-pass filter corresponding to $R(\kappa)$ for $\kappa\in[0, 10]$.
}
\label{img:RFunction}
\vspace*{-1em}
\end{figure}
\fi

Comparing Algorithms \ref{alg1} and \ref{alg2}, one can discover that the major difference between the proposed EM-GS and biased GS is attributed to the ratio of Bessel functions, $R(\kappa_n^t)=\frac{I_1(\kappa_n^t)}{I_0(\kappa_n^t)}$.  According to the properties of Bessel functions~\cite{SpecialFunction_Fox2006}, $R(\kappa_n^t)$ is monotonically increasing and bounded between $[0, 1]$.  As plotted in Fig.~\ref{img:RFunction}, the shape of  $R(\kappa_n^t)$ resembles that of a \emph{high-pass filter}. Note that the variable $\kappa_n^t = \frac{2z_n|\mb{a}_n^H\mb{s}^{t-1} + b_n|}{\sigma^2}$ implicitly represents the SNR at the $n$-th antenna. A stronger received signal energy $z_n|\mb{a}_n^H\mb{s}^{t-1} + b_n|$ together with a smaller noise power $\sigma^2$ leads to a larger value of $\kappa_n^t$. Therefore, by multiplying $z_n$ with $R(\kappa_n^t)$, we can filter out the received signals $z_n$ that have a low SNR, while keeping those $z_n$ that have a large SNR and contain useful information about the symbol $\mb{s}$. Consequently, the high-pass nature of $R(\kappa_n^t)$ ensures that the EM-GS outperforms biased GS in the low SNR regime, while the EM-GS will automatically degenerate to biased GS at high SNRs, where $R(\kappa_n^t) \rightarrow 1$ when $\kappa_n^t \rightarrow +\infty$. Simulation results are provided in Section \ref{sec:4} to quantitatively compare the performance of biased GS and EM-GS algorithms.

\subsection{Performance Analysis}\label{sec:3E}
\subsubsection{CRLB Analysis}
We now derive the CRLB for evaluating the accuracy of the detected $\tilde{\mb{s}}$ from Algorithms 1 and 2. CRLB provides a lower bound of $\mathsf{E}(\|\mb{s} - \tilde{\mb{s}}\|_2^2)$ for any unbiased estimate. Specifically, the CRLB of the distribution $p(\mb{z}; \mb{s}) = \prod_{n = 1}^N p(z_n; \mb{s})$ is expressed as
\begin{align} \label{eq:CRLB}
     \mb{C}_{\mb{s}} = \mathsf{E}(\mb{s} - \tilde{\mb{s}})(\mb{s} - \tilde{\mb{s}})^H  \ge \mb{I}^{-1},
\end{align}
where $\mb{I}$ denotes the Fisher information matrix, the $(p,q)$-th entry of which is determined by 
$
    I_{pq} = -\sum_{n=1}^N\mathsf{E}\left(\frac{\partial^2}{\partial s_p^* \partial s_q} \log p(z_n;\mb{s})\right)$. 
An explicit expression of $\mb{I}$ is provided by the following lemma. 
\begin{lemma}\label{lemma2}
    \emph{The Fisher information matrix for the Rician distribution \eqref{eq:Rician} is given by 
    \begin{align}
        \mb{I} = \sum_{n = 1}^N \beta_n \mb{a}_n \mb{a}_n^H, 
    \end{align}
    where 
    $\beta_n = \frac{1}{\sigma^4}(\mathsf{E}\{z_n^2R^2(\kappa_n)\} - |\lambda_n|^2)$.}
\end{lemma}
\begin{proof}
    (See Appendix C).
\end{proof}

As a result, the mean square error $\mathsf{E}(\|\mb{s} -\tilde{\mb{s}}\|_2^2)$ is lower bounded by $\mathsf{Tr}(\mb{I}^{-1})$, 
where $\mathsf{Tr}(\cdot)$ denotes the trace operator. In Section \ref{sec:5}, we  numerically show that our proposed algorithms can approach the above CRLB. 




\subsubsection{Computational Complexity Analysis}
It is of interest to compare the computational complexity of the biased GS, EM-GS, and exhaustive search algorithms. Compared to the biased GS algorithm, the EM-GS algorithm only needs the additional computing of $R(\kappa_n^t)$, which can efficiently rely on searching over a lookup table. Hence, these two algorithms share the same order of complexity. Take the EM-GS algorithm as an example. Its complexity is mainly attributed to the initialization and iteration operations. The former requires the eigenvalue decomposition of $\mb{M}\in\mathbb{C}^{K\times K}$, which has a complexity of $\mathcal{O}(K^3)$. Moreover, the complexity of solving the linear regression problem in \eqref{eq:LR} with $t_0$ iterations is given as $\mathcal{O}(t_0NK^2)$. As a result, the computational complexity of both the biased GS and EM-GS algorithms is $\mathcal{O}(t_0NK^2 + K^3)$.

The exhaustive search method has to try all possible symbol combinations for all users from the QAM constellation, $\mathcal{S}$, to optimize \eqref{eq:LS} and \eqref{eq:ML}. Since the number of combinations is $|\mathcal{S}|^K$, the complexity of exhaustive search is $\mathcal{O}(|\mathcal{S}|^K NK)$.  It is clear that the biased GS and EM-GS algorithms have much lower computational complexity than an exhaustive search. 

{\color{black}
\section{Extension to Channel Estimation}
In previous discussions, we assume that the channel state information (CSI) is perfectly known to the receiver. 
In case the perfect CSI is not available, we can estimate the CSI, $\mathbf{A}$, by sending pilot signals to the receiver. 

Consider the received signal at the $n$-th antenna: $z_n = |\mathbf{a}_n^H \mathbf{s} + b_n + w_n| = |\mathbf{s}^H \mathbf{a}_n + b_n^* + w_n^*|$, where $\mathbf{a}_n$ is the channel vector to be estimated. In practice, the LO is deployed at a fixed position near the atomic receiver. Thus, the coefficient $b_n$ almost keeps unchanged and can be obtained in advance by measuring the physical locations of LO and receiver. To estimate the channel coefficients $\mathbf{a}_n$, we let the $K$ users send known pilot signals during $P$ time slots to the receiver, which are denoted as $\mathbf{s}_p\in\mathbb{C}^{K\times 1}, p\in \{1,\cdots,P\}$. The overall received signal $\mathbf{z}_n \in \mathbb{R}^{P\times 1}$ is thus expressed as 
\begin{align}\label{eq:CE}
    \mathbf{z}_n &= |[\mathbf{s}_1, \cdots,\mathbf{s}_P]^H \mathbf{a}_n + [b_n, \cdots,b_n]^H + [w_{n,1}, \cdots, w_{n,P}]^H| \notag \\
    & =|\mathbf{S}^H \mathbf{a}_n + \mathbf{b}_n + \mathbf{w}_n|,
\end{align}
where $\mathbf{S} = [\mathbf{s}_1, \dots,\mathbf{s}_P] \in \mathbb{C}^{K\times P}$ refers to the known pilot matrix, $\mathbf{b}_n = [b_n,\cdots,b_n]^H \in \mathbb{C}^{P\times 1}$ the reference signal, and $\mathbf{w}_n = [w_{n,1}, \cdots, w_{n,P}]^H \in \mathbb{C}^{P\times 1}$ the noise. It is clear that the channel estimation model in \eqref{eq:CE} for recovering $\mb{a}_n$ from $\mb{z}_n$ has the same mathematical structure as the signal detection model $\mb{z} = |\mb{A}^H\mb{s} + \mb{b} + \mb{w}|$ for detecting $\mb{s}$ from $\mb{z}$. Thereafter, both the biased GS and EM-GS algorithm can be employed to estimate the channel coefficients $\mb{a}_n$. In the end, this process can be simultaneously carried out on all receive antennas to recover the whole channel matrix $\mb{A} = [\mb{a}_1, \mb{a}_2, \cdots, \mb{a}_N]$. This completes the channel estimation task.
}

\section{Simulation Results}\label{sec:5}
\subsection{Experimental Settings}

\begin{table}[t]
    \centering
    \caption{Simulation Parameters of Channel Model}
    \begin{tabular}{|c|c|}
    \hline
    { \textbf{Channel parameters}} & {\textbf{Values}}                   \\ \hline
    { Number of clusters}          & {23}                                \\ \hline
    { Number of paths per cluster}  & {20}                                \\ \hline
    { Path gains}                  & {$\mathcal{CN}(0,1)$} \\ \hline
    { Incident angles}             & {$\mathcal{U}(-90^{\circ}, 90^{\circ})$} \\ \hline
    { Maximum angle spread per cluster}                & {$\mathcal{U}(-5^{\circ}, 5^{\circ})$}                \\ \hline
    { Maximum delay spread}                & {$\mathcal{U}(0\:\text{ns}, 30\:\text{ns})$}                \\ \hline
    \end{tabular}\label{tab1}
    \vspace*{-1em}
\end{table}

The default simulation setups for atomic MIMO receivers are as follows unless specified otherwise.  
The number of atomic antennas is $N = 36$, while the number of single-antenna users is $K = 3$. 
The 4-QAM and 16-QAM modulators are considered. 
The Rydberg energy levels $52D_{5/2}$ and $53 P_{3/2}$ are adopted for detecting the sub-6G signals of frequency $\omega_{eg} = 2\pi \times 5$ GHz. Utilizing the Python package~\cite{SIBALIC2017319}, the transition dipole moment $\boldsymbol{\mu}_{eg}$ over the states $52D_{5/2}$ and $53 P_{3/2}$ is calculated as $[0, 1785.916\:qa_0, 0]^T$, 
where $a_0 = 5.292 \times 10^{-11}$ m specifies the Bohr radius. 
Taking into account the randomness of the polarization direction, vectors $\boldsymbol{\epsilon}_{nkl}$ and $\boldsymbol{\epsilon}_{b,n}$ are randomly sampled from unit circles perpendicular to their incident angles. 
In addition, the channel coefficients are generated using the standard 3GPP TR 38.901 model, whose key parameters are given in Table \ref{tab1}. 
The received SNR is defined as 
\begin{align}
{\rm SNR} = \frac{\mathsf{E}(|\mb{a}_n^H\mb{s}|^2)}{\mathsf{E}(|w_n|^2)}.
\end{align} 
Moreover, we define the \emph{reference-to-signal ratio} (RSR) as 
\begin{equation}
    {\rm RSR} = \frac{\mathsf{E}(|b_n|^2) }{\mathsf{E}(|a_{nk}s_k|^2)},
\end{equation}
which accounts for the relative intensity of the reference source. 
In our simulation, the SNR is varying from $-5\:{\rm dB}$ to $12\:{\rm dB}$ and  RSR grows from $0\:{\rm dB}$ to $25\:{\rm dB}$.   The number of iterations $t_0$ is set as 50 for both  Algorithms \ref{alg1} and \ref{alg2}.

Five benchmarking schemes are considered in performance comparison with the proposed biased GS and EM-GS algorithms. 
\begin{itemize}
    \item \emph{ZF with known phase}: Assume the ideal case where the true phase of $\mb{y}$ is known in advance. Then the \emph{zero-forcing} (ZF) detector is employed to recover the QAM symbol, $\mb{s}$. 
    \item \emph{CRLB}: We use numerical integration to calculate the normalized CRLB, $\frac{\mathsf{Tr}(\mb{I}^{-1})}{\mathsf{E}(\|\mb{s}\|_2^2)} = \frac{1}{K}\mathsf{Tr}(\mb{I}^{-1})$. This benchmark is only used for evaluating the NMSE performance, since the normalized CRLB can be regarded as a NMSE lower bound of all PR solvers. 
    \item \emph{Exhaustive search (LS)}: This method, employed for evaluating the BER, exhaustively searches all feasible constellation points to solve the LS problem in \eqref{eq:LS}.
    \item \emph{Exhaustive search (ML)}: This scheme is similar to the preceding one but aims to solve the ML problem in \eqref{eq:ML}.
    \item \emph{CM-ZF}: 
    \color{black} This method is proposed to infer amplitude modulation symbols 
    from \emph{channel magnitude} (CM)~\cite{CM_Yue2019}. It utilizes an 
    expanded linear MIMO model to approximate the PR model in \eqref{eq:z=|As|} 
    and addresses it via a ZF detector. 
    For a fair comparison, we extend this approximation to the biased PR problem and apply the idea of CM-ZF algorithm to recover QAM symbols.
\end{itemize}

\subsection{Evaluation of NMSE Performance}

\ifx\onecol\undefined
\begin{figure}
\centering
\includegraphics[width=3.2in]{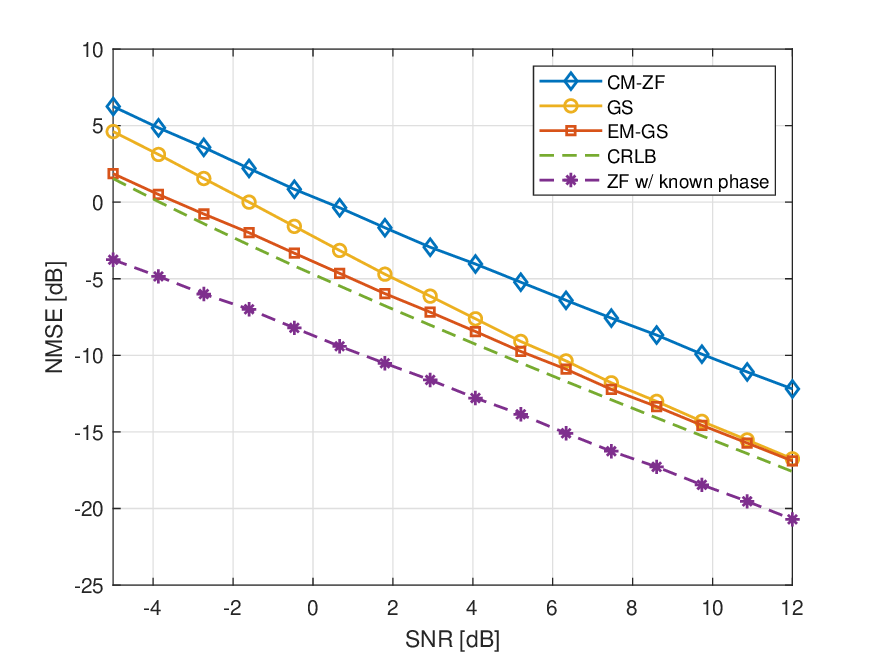}
\vspace*{-1em}
\caption{The dependence of NMSE performance on SNR for a 16-QAM modulator under $12\:{\rm dB}$ RSR.}
\label{img:nmse_snr}
\vspace*{-1em}
\end{figure}
\else
\begin{figure}
\centering
\includegraphics[width=3.4in]{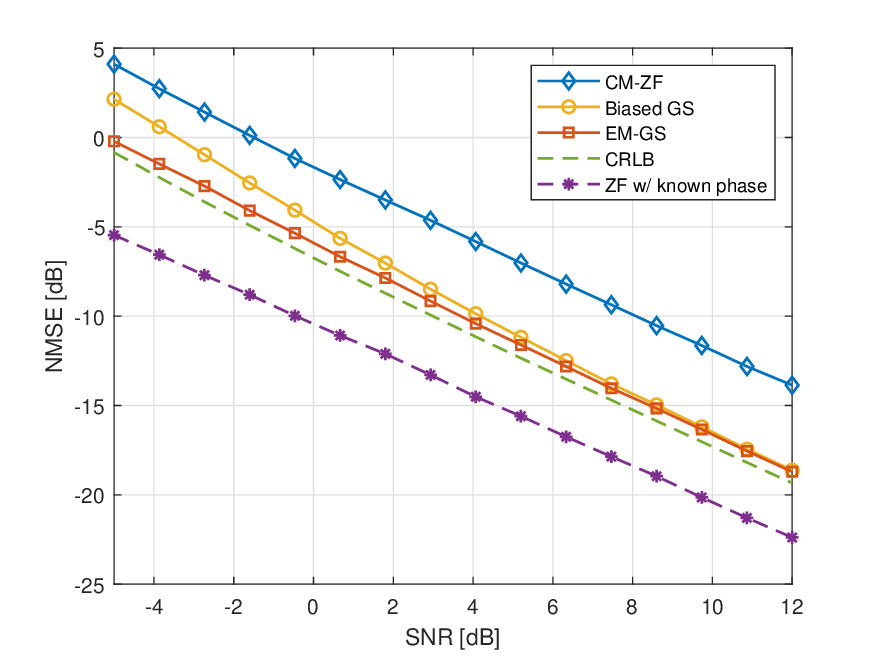}
\vspace*{-1em}
\caption{The dependence of NMSE performance on SNR for a 16-QAM modulator under $12\:{\rm dB}$ RSR.}
\label{img:nmse_snr}
\vspace*{-1em}
\end{figure}
\fi

Define the NMSE as  $\frac{\mathsf{E}(\|\mb{s} -\tilde{\mb{s}}\|^2_2)}{\mathsf{E}(\|\mb{s}\|_2^2)} = \frac{1}{K} \mathsf{E}(\|\mb{s} -\tilde{\mb{s}}\|^2_2)$, where $\tilde{\mb{s}}$ represents the recovered signal without constellation demapping, e.g., the direct outputs of Algorithms 1 and 2. 
The curves of NMSE versus SNR are plotted in Fig.~\ref{img:nmse_snr}. 
The RSR is fixed as $12\:{\rm dB}$ and the 16-QAM modulator is employed. 
It is observed from Fig.~\ref{img:nmse_snr} that both the biased GS and the proposed EM-GS algorithms outperform the CM-ZF by $2\sim 5\:{\rm dB}$ in NMSE. 
Besides, the proposed EM-GS algorithm is always better than the biased GS algorithm, especially in the low SNR regime, e.g., the performance gap is around 2 dB when ${\rm SNR} = -4\:{\rm dB}$.
This observation is consistent with the high-pass nature of the function $R(\cdot)$ we analyzed earlier. 
More importantly, the NMSE performance of EM-GS is close to CRLB, given that the ML estimator is an asymptotic minimum-variance unbiased estimator. 
The final interesting observation is that the NMSE gap between the ZF with known phase and  CRLB approaches $3$ dB as the SNR increases, implying that the vanished phase information imposes a $3$ dB NMSE loss to atomic receivers.

\ifx\onecol\undefined
\begin{figure}
\centering
\includegraphics[width=3.2in]{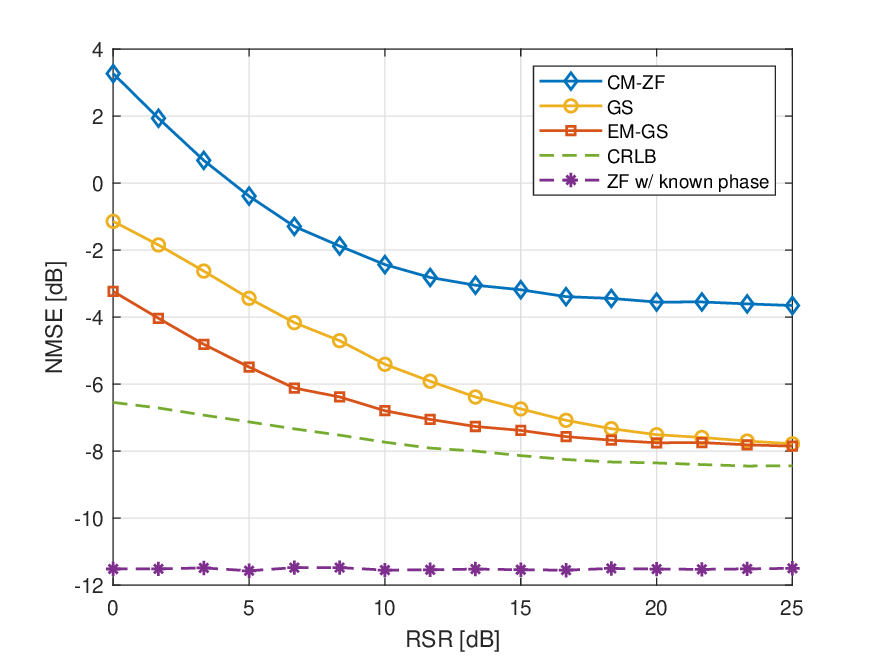}
\vspace*{-1em}
\caption{The effect of RSR on the NMSE performance for a 16-QAM modulator under $3\:{\rm dB}$ SNR.}
\label{img:nmse_rsr}
\vspace*{-1em}
\end{figure}
\else
\begin{figure}
\centering
\includegraphics[width=3.4in]{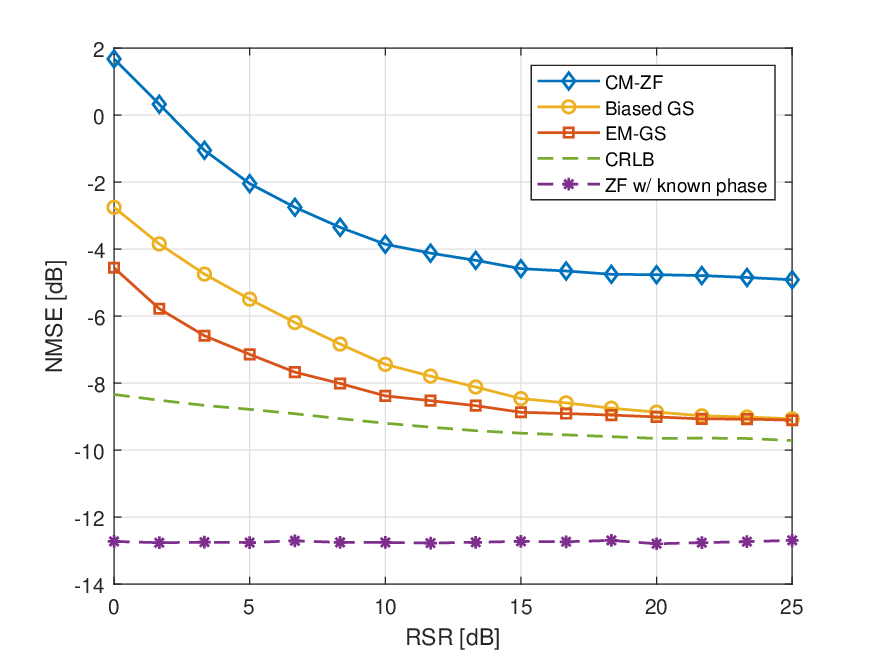}
\vspace*{-1em}
\caption{The effect of RSR on the NMSE performance for a 16-QAM modulator under $3\:{\rm dB}$ SNR.}
\label{img:nmse_rsr}
\vspace*{-1em}
\end{figure}
\fi

We then investigate the effect of RSR on NMSE as shown in Fig.~\ref{img:nmse_rsr}, where the SNR is fixed as $3$ dB and the 16-QAM modulator is adopted. 
As opposed to a uniform NMSE performance achieved by the ZF with known phase, the NMSE of all PR solvers rapidly declines by $5$ dB as the RSR increases. 
This result is surprising but still reasonable because the atomic MIMO receiver follows a non-linear model (\ref{eq:BiasedPR}).  
A stronger reference signal makes it easier for PR solvers to estimate the phase difference between $\mb{s}$ and $s_b$, giving rise to a more accurate detection of the true phase of $\mb{s}$. 
Furthermore, the large RSR situation is common in practice because a LO can be deployed near the atomic receiver. 
Taking account of the path loss that decays quadratically with the path length, the strength of the reference signal can be several times greater than the user signal. 
For example, if the LO-to-atomic receiver distance is shorter than a quarter of the user-to-atomic receiver distance, the RSR is greater than $10\log_{10}4^2\approx 12$ dB.

\ifx\onecol\undefined
\begin{figure}
    \centering
    \subfigure[4-QAM, $N\times K  = 36\times 3$]{\includegraphics[width=3.2in]{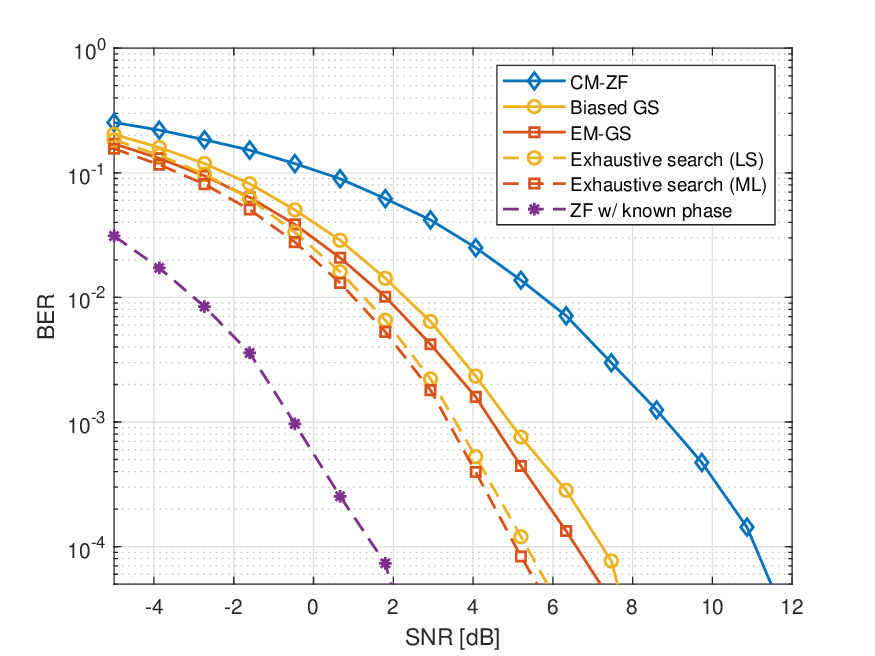}}\\
    \vspace*{-1em}	
    \subfigure[16-QAM, $N \times K = 100 \times 6$]{\includegraphics[width=3.2in]{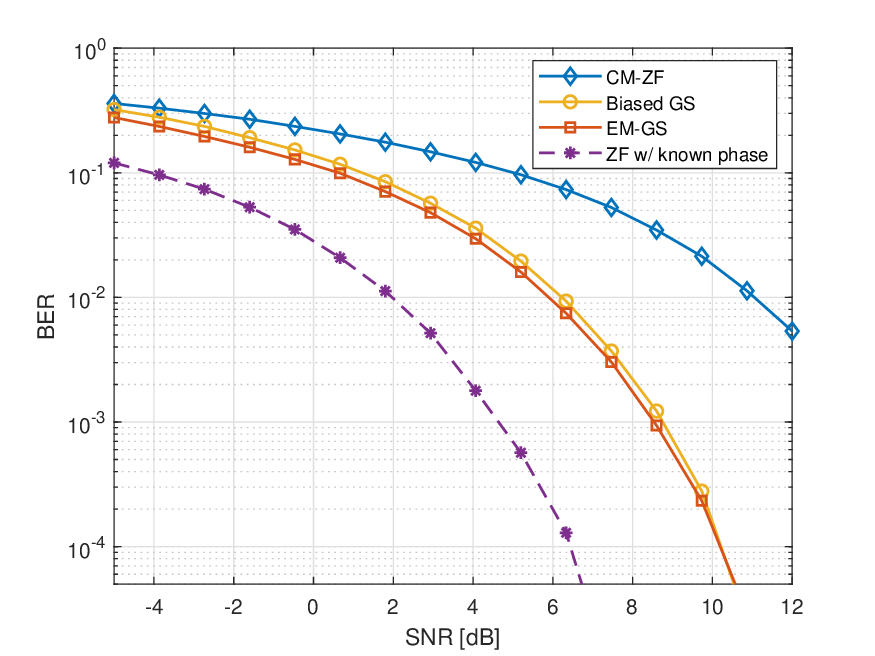}} 
    \caption{The effect of SNR on the BER performance under (a) the small-scale configuration (i.e., $N\times K = 36\times 3$, 4-QAM modulator) and (b) the large-scale configuration (i.e., $N\times K = 100\times 6$, 16-QAM modulator).} 
	\vspace*{-1em}
	\label{img:ber_snr}
\end{figure}
\else 
\begin{figure*}
	\centering
	\subfigure[4-QAM, $N\times K  = 36\times 3$]{
		\begin{minipage}[t]{0.49\linewidth}
			\centering
			\includegraphics[width=3in]{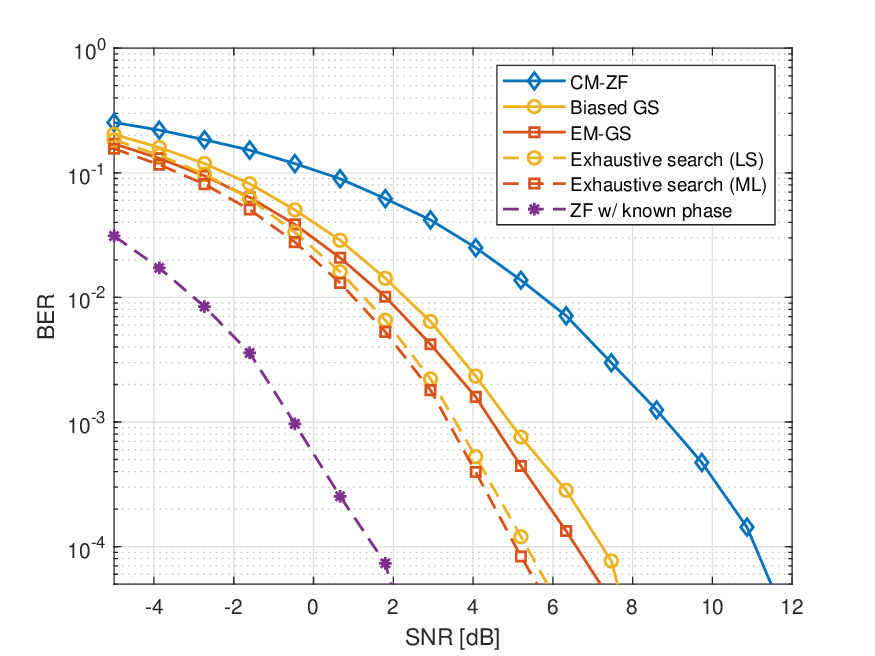}\\
			\vspace{0.02cm}
		\end{minipage}%
	}%
	\subfigure[4-QAM, $N\times K  = 36\times 3$]{
		\begin{minipage}[t]{0.49\linewidth}
			\centering
			\includegraphics[width=3in]{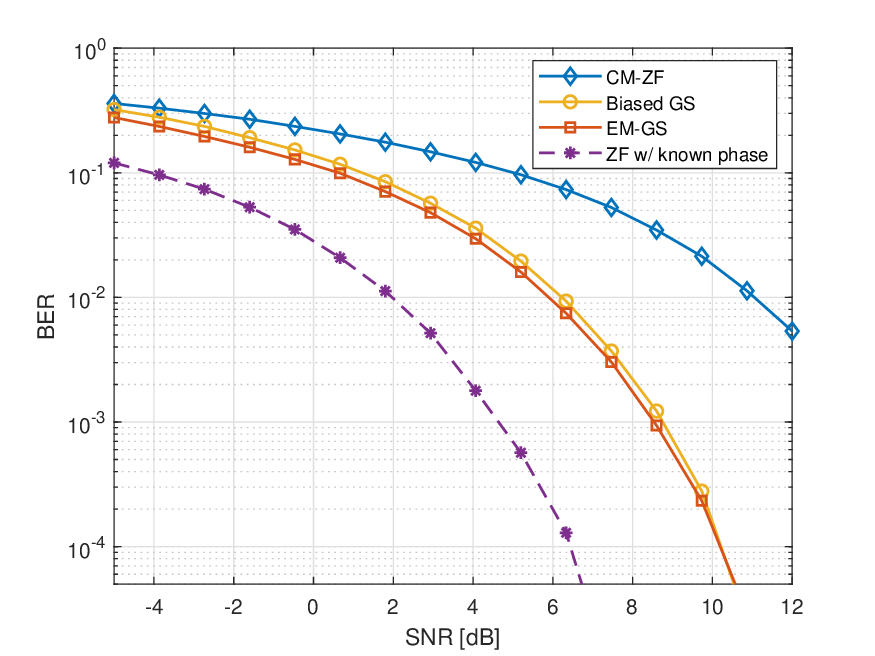}\\
			\vspace{0.02cm}
		\end{minipage}%
	}%
	\centering
	\caption{The effect of SNR on the BER performance under (a) the small-scale configuration (i.e., $N\times K = 36\times 3$, 4-QAM modulator) and (b) the large-scale configuration (i.e., $N\times K = 100\times 6$, 16-QAM modulator).}
	\vspace{-0.2cm}
	\label{img:ber_snr}
\end{figure*}
\fi


\subsection{Evaluation of BER Performance}
To evaluate the BER performance, the constellation demapping step is introduced to project the recovered signals $\tilde{s}_k$ for each user to the nearest constellation point $\hat{s}_k$ and then to 0-1 bits. 
The curves of BER versus SNR are plotted in Fig.~\ref{img:ber_snr}. 
The RSR is set as $12$ dB. 
Fig.~\ref{img:ber_snr}(a) adopts a small-scale configuration: $N\times K = 36\times 3$ with a 4-QAM modulator, while Fig.~\ref{img:ber_snr}(b) employs a large-scale configuration: $N\times K = 100 \times 6$ with a 16-QAM modulator. It can be seen from Fig.~\ref{img:ber_snr}(a) that the biased GS and the EM-GS algorithms exhibit a remarkable BER reduction compared with CM-ZF and can perform similarly as the exhaustive search method with much lower complexity. 
Furthermore, the EM-GS algorithm consistently outperforms the biased GS algorithm due to the high-pass filter $R(\cdot)$ (see their comparison in Section \ref{sec:4C}). 
For the large-scale configuration depicted in Fig.~\ref{img:ber_snr}(b), the computation of exhaustive search method is prohibitive, so it is excluded from this comparison. 
We can observe that the SNR gap between the EM-GS algorithm and the ZF method with known phase for realizing the same BER level is between $3\sim4$ dB. 
This fact indicates that despite the lack of phase information, our EM-GS based detector can attain a similar BER trend as the ZF detector in the ideal case, demonstrating the effectiveness of multi-user atomic-MIMO detection.

Finally, the influence of RSR on BER performance is illustrated in Fig.~\ref{img:ber_rsr}. 
The SNR is fixed as $3$ dB and a 4-QAM modulator is adopted. 
Similar to Fig.~\ref{img:nmse_rsr}, one can observe that the increase in RSR can considerably reduce the BER of all PR solvers. 
Take the proposed EM-GS as an example. More than one order of magnitude reduction in BER is achievable by increasing the RSR from $0$ dB to $20$ dB.  
Therefore, we can draw the conclusion that a strong received reference signal is necessary for detecting QAM symbols by atomic receivers.


\ifx\onecol\undefined
\begin{figure}
\centering
\includegraphics[width=3.4in]{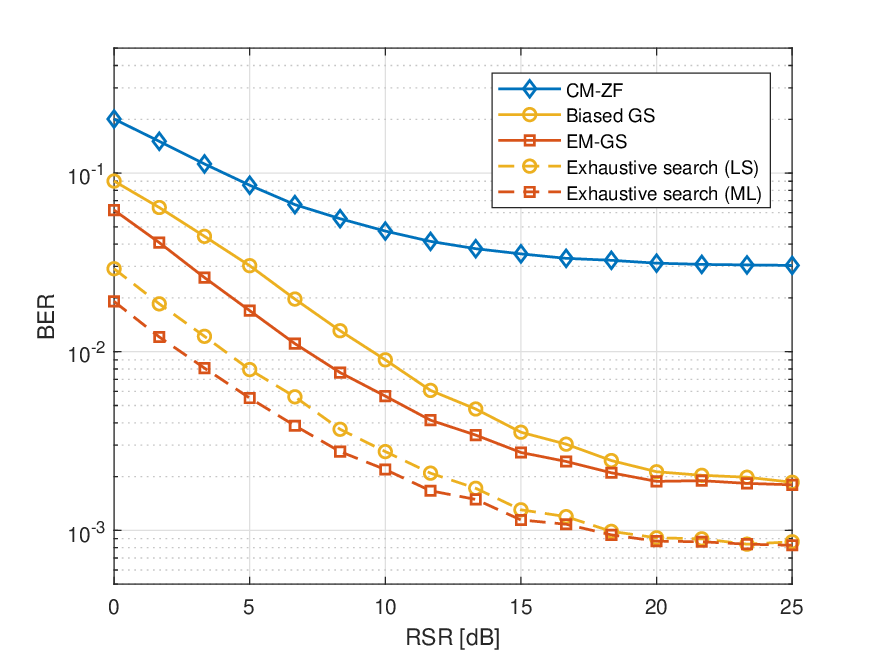}
\vspace*{-1em}
\caption{The effect of RSR on the BER performance under different detection algorithms for a 16-QAM modulator under $3\:{\rm dB}$ SNR.}
\label{img:ber_rsr}
\vspace*{-1em}
\end{figure}
\else
\begin{figure}
\centering
\includegraphics[width=3.4in]{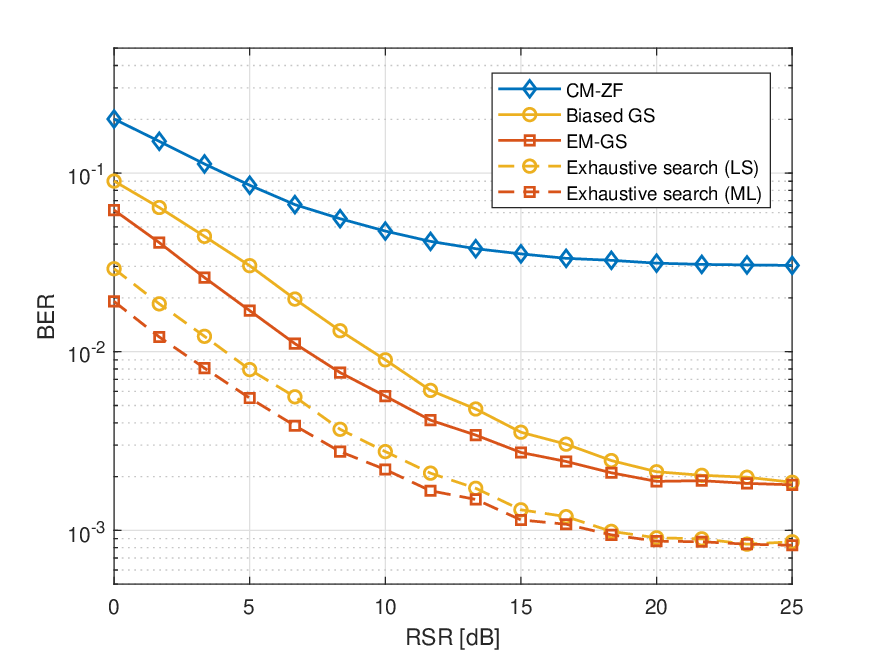}
\vspace*{-1em}
\caption{The effect of RSR on the BER performance achieved by different detection algorithms for a 16-QAM modulator under $3\:{\rm dB}$ SNR.}
\label{img:ber_rsr}
\vspace*{-1em}
\end{figure}
\fi

\section{Conclusions} \label{sec:6}
In this paper, we have demonstrated the feasibility of atomic MIMO receivers, marking the first attempt to introduce atomic receivers into MIMO communications. 
Different from the classical linear MIMO model, the signal detection of atomic MIMO receiver is shown to be a non-linear biased PR problem. 
Accordingly, two algorithms, the biased GS and the EM-GS, are proposed to detect symbols based on the LS and ML criteria, respectively, which are verified to be near-optimal. 
Comprehensive experiment results validate the efficiency and effectiveness of atomic MIMO receiver. 

This work serves as an important step towards advanced atomic wireless receivers for next-generation communication systems. 
Several unexplained issues warrant follow-up studies, such as theoretically determining the channel capacity of atomic receivers, given that the transmission model is no longer linear and Gaussian.
Furthermore, integrating atomic receivers with various modern communication techniques, such as wideband, cell-free, mmWave/THz, and RIS-aided communications, and developing such receivers to enable over-the-air computation, edge learning, and ISAC, are also of interest for 6G research. 
Furthermore, exploring the potential incorporation of more quantum information technologies into atomic receivers to achieve unprecedented communication capabilities presents a promising research direction.
As a preliminary study, this work is expected to inspire more innovations in the development of atomic receivers for advancing wireless communications.

\appendix
{
\color{black}

\subsection{Rotating Wave Approximation for Solving  \eqref{eq:scheqn}}
We follow the standard RWA to solve \eqref{eq:scheqn}. For ease of notations, the term ``$(t)$" in variables $\alpha_{e,n}(t)$, $\alpha_{g,n}(t)$, and $\xi_n(t)$ are neglected temporarily.
The Schrödinger equation in \eqref{eq:scheqn} is explicitly written as 
\begin{align}\label{eq:A1}
i\hbar \left(
    \begin{array}{c}
          \dot{\alpha}_{e,n}\\
         \dot{\alpha}_{g,n}
    \end{array} \right) = 
    \left(
    \begin{array}{cc}
          \hbar\omega_e & \xi_n\\
         \xi_n^* & \hbar \omega_g
    \end{array} \right)\left(
    \begin{array}{c}
          {\alpha}_{e,n}\\
         {\alpha}_{g,n}
    \end{array} \right)
    ,
\end{align}
where $\dot{\alpha}$ denotes the time derivative of $\alpha$. 
To solve this time-varying differential equation, the first step is to switch \eqref{eq:A1} into the interaction picture of signal frequency. Specifically, let $\widetilde{\alpha}_{e,n} = e^{i(\omega + \omega_g)t}\alpha_{e,n}$ and $\widetilde{\alpha}_{g,n} = e^{i\omega_gt}\alpha_{g,n}$. By applying the relationships in \eqref{eq:A1}, the differential equation of $\widetilde{\alpha}_{e,n}$ is derived as 
\begin{align}\label{eq:A2}
 i\hbar \dot{\widetilde{\alpha}}_{e,n} &= e^{i(\omega + \omega_g)t}  i\hbar \dot{\alpha}_{e,n} - \hbar(\omega + \omega_g)e^{i(\omega + \omega_g)t}\alpha_{e,n}\notag \\
 &=e^{i(\omega + \omega_g)t}(\hbar \omega_e  \alpha_{e,n} + \xi_n \alpha_{g,n}) - \hbar(\omega + \omega_g){\widetilde{\alpha}}_{e,n} \notag \\
 &=\hbar \delta {\widetilde{\alpha}}_{e,n} + e^{i\omega t}\xi_n \widetilde{\alpha}_{g,n},
\end{align}
and the differential equation of $\widetilde{\alpha}_{g,n}$ is similarly derived as
\begin{align}\label{eq:A3}
    i\hbar \dot{\widetilde{\alpha}}_{g,n} &= e^{i\omega_gt}  i\hbar \dot{\alpha}_{g,n} - \hbar\omega_g{\widetilde{\alpha}}_{g,n}  = e^{-i\omega t}\xi_n^* \widetilde{\alpha}_{e,n}.
\end{align}
Given the form of $\xi_n$ in \eqref{eq:xin}, the time-dependent coefficient $e^{i\omega t}\xi_n$ in \eqref{eq:A2} and \eqref{eq:A3} is written as 
\begin{align}
    e^{i\omega t}\xi_n &= \xi_n^{-} + \xi_n^{+}, 
\end{align}
where $\xi_n^{-} = \frac{1}{2}\sum_{k}\sum_{l}{\boldsymbol{\mu}}_{eg}^T\boldsymbol{\epsilon}_{nkl}  \sqrt{P_k}\rho_{nkl} |s_k| e ^{-j(\phi_{nkl} + \gamma_k)}$ is a time-independent term while the time-dependent term $\xi_n^{+} = \frac{1}{2}\sum_{k}\sum_{l}{\boldsymbol{\mu}}_{eg}^T\boldsymbol{\epsilon}_{nkl}  \sqrt{P_k}\rho_{nkl} |s_k| e ^{j(\phi_{nkl} + \gamma_k)} e^{j2\omega t}$ oscillates with frequency $2\omega$. \emph{Due to the rapid oscillation of $\xi_n^{+}$, the RWA states that the time integrals of $\xi_n^{+}\widetilde{\alpha}_{g,n}$ and ${\xi_n^{+}}^*\widetilde{\alpha}_{e,n}$ quickly averages to zero, and thus can be neglected.} Thereafter, the equations \eqref{eq:A2} and \eqref{eq:A3} are approximated as 
\begin{align}\label{eq:A5}
i\hbar \left(
    \begin{array}{c}
          \dot{\widetilde{\alpha}}_{e,n}\\
         \dot{\widetilde{\alpha}}_{g,n}
    \end{array} \right) = 
    \left(
    \begin{array}{cc}
          \hbar\delta & \xi_n^{-}\\
         {\xi_n^{-}}^* & 0
    \end{array} \right)\left(
    \begin{array}{c}
          \widetilde{\alpha}_{e,n}\\
         \widetilde{\alpha}_{g,n}
    \end{array} \right)
    .
\end{align}
The effective Hamiltonian is thereby $\widetilde{H}_n = \hbar \delta \ket{e}\bra{e} + \xi_n^{-}\ket{e}\bra{g} + {\xi_n^{-}}^*\ket{g}\bra{e}$. 
Clearly, the RWA transforms the time-dependent Hamiltonian in \eqref{eq:A1} into a time-independent one, $\widetilde{H}$, which simplifies our problem significantly. To solve this new Schrödinger equation, \eqref{eq:A5}, we need to find the eigenvector decomposition (EVD) of $\widetilde{H}_n$, denoted as $\widetilde{H}_n = E_{n}^+\ket{E_n^+}\bra{E_n^{+}} + E_n^-\ket{E_n^-}\bra{E_n^{-}}$. Specifically, the eigenvalues are calculated as
\begin{align}
E_n^{\pm} = \frac{1}{2}\hbar \delta \pm \sqrt{|\xi_n^{-}|^2 + \frac{1}{4}\hbar^2\delta^2 }. 
\end{align}
The corresponding eigenvectors are derived as 
\begin{align}
    \left\{
    \begin{array}{c}
         \ket{E_n^{+}} = e^{-j\Psi_n}\sin(\Phi_n/2)\ket{g} + \cos(\Phi_n/2)\ket{e}  \\
         \ket{E_n^{-}} = e^{-j\Psi_n}\cos(\Phi_n/2)\ket{g} - \sin(\Phi_n/2)\ket{e} 
    \end{array}
    \right. ,
\end{align}
where $\Phi_n = \arctan(\frac{2|\xi_n^{-}|}{\hbar\delta})$ and $\Psi_n = \angle(\xi_n^{-})$. 
Thereafter, the evolution of $(\widetilde{\alpha}_{e,n},
         \widetilde{\alpha}_{g,n})$ driven by \eqref{eq:A5} is expressed as
 \begin{align}
  (\widetilde{\alpha}_{e,n},
         \widetilde{\alpha}_{g,n})^T = c_{n}^+ e^{-i \frac{E_n^+}{\hbar} t} \ket{E_n^{+}} + c_{n}^- e^{-i \frac{E_n^-}{\hbar} t} \ket{E_n^{-}}.
 \end{align}
Here, the coefficients $c_{n}^{\pm}$ are determined by the initial condition. Given the initial states $\widetilde{\alpha}_{e, n}(0) = \alpha_{e, n}(0) = 0$ and $\widetilde{\alpha}_{g,n}(0) = {\alpha}_{g,n}(0) = 1$, 
associated with the initial derivative $\dot{\widetilde{\alpha}}_{e, n}(0) = -i\delta \widetilde{\alpha}_{e, n}(0) -i\frac{\xi_n^{-}}{\hbar}\widetilde{\alpha}_{g,n}(0) = -i\frac{\xi_n^{-}}{\hbar}$, 
the probability density is obtained as 
\begin{align}
    |\alpha_{e, n}(t)|^2 = |\widetilde{\alpha}_{e, n}(t)|^2 = \frac{\Omega_{R, n}^2}{\Omega_{n}^2} \sin^2\left(\frac{\Omega_nt}{2}\right), 
\end{align}
where  $\Omega_{R,n} = 2|\xi_n^-|/\hbar = 2|{\xi_n^-}^*|/\hbar$ and $\Omega_n = \sqrt{\Omega_{R,n}^2 + \delta^2}$. Last, owing to the the resonance condition that $\delta = 0$, we arrive at \eqref{eq:RabiMIMO}. 

}

\subsection{Proof of Lemma 1} 

  Substituting \eqref{eq:Gaussian} and \eqref{eq:von Mises} into \eqref{eq:surrogate}, the surrogate function can be expressed as
  \ifx\onecol\undefined
  \begin{align}\label{eq:L1P1}
      Q(\mb{s}|\mb{s}^{t-1})
      & = \sum_{n = 1}^N
      \int_0^{2\pi} p(\theta_n|z_n, \mb{s}^{t - 1}) \log p(z_n, \theta_n; \mb{s}) \text{d}\theta_n  \notag\\
      &\overset{(a)}{=} - \frac{1}{\sigma^2}\sum_{n = 1}^N (|\lambda_n|^2 - 2|\lambda_n| z_n L_n) + C, 
  \end{align}
  \else 
  \begin{align}\label{eq:L1P1}
      Q(\mb{s}|\mb{s}^{t-1})
      & = \sum_{n = 1}^N
      \int_0^{2\pi} p(\theta_n|z_n, \mb{s}^{t - 1}) \log p(z_n, \theta_n; \mb{s}) \text{d}\theta_n  \overset{(a)}{=} C- \frac{1}{\sigma^2}\sum_{n = 1}^N (|\lambda_n|^2 - 2|\lambda_n| z_n L_n) , 
  \end{align}
  \fi 
where $L_n = \frac{1}{2\pi I_0(\kappa_n^t)}\int_0^{2\pi} e^{\kappa_n^t\cos(\theta_n - \theta_n^t)} \cos(\theta_n - \angle\lambda_n) \text{d}\theta_n$, and (a) is derived by discarding all terms independent of $\mb{s}$ into $C$. By changing the integral variable from $\theta_n$ to $\theta_n - \theta_n^t$ and applying the periodic property of sinusoidal functions, the integral $L_n$  can be presented as 
\ifx\onecol\undefined
\begin{align} \label{eq:L1P2}
    L_n 
    &= \frac{1}{2\pi I_0(\kappa_n^t)}\int_0^{2\pi} e^{\kappa_n^t\cos\theta_n} \cos(\theta_n   + \theta_n^t - \angle\lambda_n)  \text{d}\theta_n \notag\\
    &\overset{(b)}{=}\frac{1}{2\pi I_0(\kappa_n^t)}\int_0^{2\pi} e^{\kappa_n^t\cos\theta_n}\cos\theta_n\cos(\theta_n^t - \angle\lambda_n)  \text{d}\theta_n \notag\\&\overset{(c)}{=}R(\kappa_n^t)\cos(\theta_n^t - \angle\lambda_n),
\end{align}
\else 
\begin{align} \label{eq:L1P2}
    L_n 
    &= \frac{1}{2\pi I_0(\kappa_n^t)}\int_0^{2\pi} e^{\kappa_n^t\cos\theta_n} \cos(\theta_n   + \theta_n^t - \angle\lambda_n)  \text{d}\theta_n \notag\\
    &\overset{(b)}{=}\frac{1}{2\pi I_0(\kappa_n^t)}\int_0^{2\pi} e^{\kappa_n^t\cos\theta_n}\cos\theta_n\cos(\theta_n^t - \angle\lambda_n)  \text{d}\theta_n \overset{(c)}{=}R(\kappa_n^t)\cos(\theta_n^t - \angle\lambda_n),
\end{align}
\fi
where (b) holds due to $\int_0^{2\pi} e^{\kappa_n^t\cos\theta_n}\sin\theta_n \text{d}\theta_n = 0$ and (c) holds because of the definition of the first-order modified Bessel function $I_1(x) = \frac{1}{2\pi} \int_0^{2\pi} e^{x\cos\theta}\cos\theta\text{d}\theta$. Finally, by bringing \eqref{eq:L1P2} back to \eqref{eq:L1P1}, we can conclude that 
\ifx\onecol\undefined
\begin{align}
&Q(\mb{s}|\mb{s}^t) \notag \\&=  -\frac{1}{\sigma^2}\sum_{n=1}^N\left(
|\lambda_n|^2 - 2|\lambda_n| z_nR(\kappa_n^t)\cos(\theta_n^t - \angle\lambda_n)
\right) + C, \notag\\
&= -\frac{1}{\sigma^2}\sum_{n=1}^N\| z_n e^{i\theta_n^t}R(\kappa_n^t)- \lambda_n\|_2^2 +C, 
\end{align}
\else 
\begin{align}
Q(\mb{s}|\mb{s}^t) &=  C-\frac{1}{\sigma^2}\sum_{n=1}^N\left(
|\lambda_n|^2 - 2|\lambda_n| z_nR(\kappa_n^t)\cos(\theta_n^t - \angle\lambda_n)
\right)   \notag\\
&= C -\frac{1}{\sigma^2}\sum_{n=1}^N\| z_n e^{i\theta_n^t}R(\kappa_n^t)- \lambda_n\|_2^2 = C- \frac{1}{\sigma^2}\| \mb{z}\circ e^{i\boldsymbol{\theta}^t} \circ R(\boldsymbol{\kappa}^t) - \mb{A}^H\mb{s} - \mb{b} \|_2^2 ,
\end{align}
\fi
 which naturally results in the conclusion in Lemma 1.

\subsection{Proof of Lemma 2}
The first and second derivatives of $\log p(z_n; \mb{s})$ are given by
\begin{align}\label{eq:d1}
    \frac{\partial \log p(z_n; \mb{s})}{\partial s_p^*} = \frac{1}{\sigma^2}\left( \frac{z_n}{|\lambda_n|}R(\kappa_n) - 1
    \right) \frac{\partial |\lambda_n|^2}{\partial s_p^*},
\end{align}
\begin{align}\label{eq:d2}
    &\frac{\partial^2 \log p(z_n; \mb{s})}{\partial s_p^* \partial s_q} = \frac{1}{\sigma^2}\left( \frac{z_n}{|\lambda_n|}R(\kappa_n) - 1
    \right) \frac{\partial^2 |\lambda_n|^2}{\partial s_p^* \partial s_q} + \\ &\frac{1}{\sigma^2|\lambda_n|^2}  
    \left( - \frac{z_n}{|\lambda_n|}R(\kappa_n) + \frac{z_n^2 - z_n^2R^2(\kappa_n)}{\sigma^2}
    \right)
    \frac{\partial |\lambda_n|^2}{\partial s_p^*}   \frac{\partial |\lambda_n|^2}{\partial s_q}.  \notag
\end{align}
In deriving \eqref{eq:d1}, the property $I_0'(z) = I_1(z)$ is harnessed, while in deriving \eqref{eq:d1}, the  property $R'(z) = 1 - R^2(z) - \frac{1}{z}R(z)$ is employed~\cite{SpecialFunction_Fox2006}. It was proven in~\cite{PR_Zhu2023} that  $\mathsf{E}\{z_nR(\kappa_n)\} = |\lambda_n|$. Therefore, the expectation of the first derivative is 
\begin{align}
    \mathsf{E}\left(\frac{\partial}{\partial s_p^*}\log p(z_n; \mb{s})\right) = 0,
\end{align}
implying that the regularity condition of CRLB is satisfied. since the second moment of Rician distribution is $\mathsf{E}(z_n^2) = |\lambda_n|^2 + \sigma^2$, the expectation of the second derivative is given as 
\ifx\onecol\undefined
\begin{align}
        \mathsf{E}\left(\frac{\partial^2\log p(z_n; \mb{s})}{\partial s_p^* \partial s_q}\right) &= \frac{|\lambda_n|^2 - \mathsf{E}\{z_n^2R^2(\kappa_n)\}}{\sigma^4|\lambda_n|^2}  
    \frac{\partial |\lambda_n|^2}{\partial s_p^*}   \frac{\partial |\lambda_n|^2}{\partial s_q} \notag\\&= -\frac{\beta_n}{|\lambda_n|^2} \frac{\partial |\lambda_n|^2}{\partial s_p^*}   \frac{\partial |\lambda_n|^2}{\partial s_q},
\end{align}
\else 
\begin{align}
        \mathsf{E}\left(\frac{\partial^2\log p(z_n; \mb{s})}{\partial s_p^* \partial s_q}\right) &= \frac{|\lambda_n|^2 - \mathsf{E}\{z_n^2R^2(\kappa_n)\}}{\sigma^4|\lambda_n|^2}  
    \frac{\partial |\lambda_n|^2}{\partial s_p^*}   \frac{\partial |\lambda_n|^2}{\partial s_q} = -\frac{\beta_n}{|\lambda_n|^2} \frac{\partial |\lambda_n|^2}{\partial s_p^*}   \frac{\partial |\lambda_n|^2}{\partial s_q},
\end{align}
\fi
where $\beta_n \overset{\Delta}{=} \frac{\mathsf{E}\{z_n^2R^2(\kappa_n)\} - |\lambda_n|^2}{\sigma^4}$.
Then, the $(p,q)$-th entry of $\mb{I}$ is derived by $I_{pq} = \sum_{n=1}^{N} \frac{\beta_n}{|\lambda_n|^2}
    \frac{\partial |\lambda_n|^2}{\partial s_p^*}   \frac{\partial |\lambda_n|^2}{\partial s_q}$ and thus the Fisher information matrix is constructed by
\begin{align}
    \mb{I} = \sum_{n = 1}^N \frac{\beta_n}{|\lambda_n|^2} \frac{\partial |\lambda_n|^2}{\partial \mb{s}^*}   \frac{\partial |\lambda_n|^2}{\partial \mb{s}^T},
\end{align} 
where $
    \frac{\partial |\lambda_n|^2}{\partial \mb{s}^*} = \frac{\partial |\mb{a}_n^H\mb{s} + b_n|^2}{\partial \mb{s}^*} = \mb{a}_n\mb{a}_n^H\mb{s} + b_n \mb{a}_n = \lambda_n\mb{a}_n
$.
Finally, substituting $\frac{\partial |\lambda_n|^2}{\partial \mb{s}^*}$ into $\mb{I}$ completes the proof.

\bibliographystyle{IEEEtran}
\bibliography{Reference.bib}

\begin{thebibliography}{10}
\providecommand{\url}[1]{#1}
\csname url@samestyle\endcsname
\providecommand{\newblock}{\relax}
\providecommand{\bibinfo}[2]{#2}
\providecommand{\BIBentrySTDinterwordspacing}{\spaceskip=0pt\relax}
\providecommand{\BIBentryALTinterwordstretchfactor}{4}
\providecommand{\BIBentryALTinterwordspacing}{\spaceskip=\fontdimen2\font plus
\BIBentryALTinterwordstretchfactor\fontdimen3\font minus
  \fontdimen4\font\relax}
\providecommand{\BIBforeignlanguage}[2]{{%
\expandafter\ifx\csname l@#1\endcsname\relax
\typeout{** WARNING: IEEEtran.bst: No hyphenation pattern has been}%
\typeout{** loaded for the language `#1'. Using the pattern for}%
\typeout{** the default language instead.}%
\else
\language=\csname l@#1\endcsname
\fi
#2}}
\providecommand{\BIBdecl}{\relax}
\BIBdecl

\bibitem{6GReview_Zhang2019}
Z.~Zhang, Y.~Xiao, Z.~Ma, M.~Xiao, Z.~Ding, X.~Lei, G.~K. Karagiannidis, and
  P.~Fan, ``{6G} wireless networks: {Vision}, requirements, architecture, and
  key technologies,'' \emph{IEEE Veh. Technol. Mag.}, vol.~14, pp. 28--41, Sep.
  2019.

\bibitem{QuanComp_Gyongyosi2019}
L.~Gyongyosi and S.~Imre, ``A survey on quantum computing technology,''
  \emph{Comput. Sci. Rev.}, vol.~31, pp. 51--71, Feb. 2019.

\bibitem{QuanComm_Wang2022}
C.~Wang and A.~Rahman, ``Quantum-enabled {6G} wireless networks:
  {Opportunities} and challenges,'' \emph{IEEE Wireless Commun.}, vol.~29, pp.
  58--69, Feb. 2022.

\bibitem{QuanSense_Degen2017}
C.~L. Degen, F.~Reinhard, and P.~Cappellaro, ``Quantum sensing,'' \emph{Rev.
  Mod. Phys.}, vol.~89, p. 035002, Jul. 2017.

\bibitem{QuanSense_Zhang2023}
F.~Zhang, B.~Jin, Z.~Lan, Z.~Chang, D.~Zhang, Y.~Jiao, M.~Shi, and J.~Xiong,
  ``Quantum wireless sensing: Principle, design and implementation,'' in
  \emph{Proc. of the 29th Annual International Conference on Mobile Computing
  and Networking (ACM MobiCom'23)}, Oct. 2023, pp. 1--15.

\bibitem{RydMag_Fancher2021}
C.~T. Fancher, D.~R. Scherer, M.~C.~S. John, and B.~L.~S. Marlow, ``Rydberg
  atom electric field sensors for communications and sensing,'' \emph{IEEE
  Trans. Quantum Eng.}, vol.~2, no. 3501313, pp. 1--13, Mar. 2021.

\bibitem{liu_continuous-frequency_2022}
X.~Liu, K.~Liao, Z.~Zhang, H.~Tu, W.~Bian, Z.~Li, S.~Zheng, H.~Li, W.~Huang,
  H.~Yan, and S.~Zhu, ``Continuous-frequency microwave heterodyne detection in
  an atomic vapor cell,'' \emph{Phys. Rev. Appl.}, vol.~18, no.~5, p. 054003,
  Nov. 2022.

\bibitem{liao_microwave_2020}
K.~Liao, H.~Tu, S.~Yang, C.~Chen, X.~Liu, J.~Liang, X.~Zhang, H.~Yan, and
  S.~Zhu, ``Microwave electrometry via electromagnetically induced absorption
  in cold {Rydberg} atoms,'' \emph{Phys. Rev. A}, vol. 101, no.~5, p. 053432,
  May 2020.

\bibitem{RydMag_Liu2023}
B.~Liu, L.~Zhang, Z.~Liu, Z.~Deng, D.~Ding, B.~Shi, and G.~Guo, ``Electric
  field measurement and application based on {Rydberg} atoms,''
  \emph{Electromagn. Sci.}, vol.~1, no.~2, pp. 1--16, Jun. 2023.

\bibitem{zhong_polar_2023}
Y.~Wang, F.~Jia, J.~Hao, Y.~Cui, F.~Zhou, X.~Liu, J.~Mei, Y.~Yu, Y.~Liu,
  J.~Zhang, F.~Xie, and Z.~Zhong, ``Precise measurement of microwave
  polarization using a {Rydberg} atom-based mixer,'' \emph{Opt. Express},
  vol.~31, no.~6, pp. 10\,449--10\,457, Mar. 2023.

\bibitem{liu_deep_2022}
Z.~Liu, L.~Zhang, B.~Liu, Z.~Zhang, G.~Guo, D.~Ding, and B.~Shi, ``Deep
  learning enhanced {Rydberg} multifrequency microwave recognition,'' \emph{Nat
  Commun}, vol.~13, no.~1, p. 1997, Apr. 2022.

\bibitem{RydReview_Saffman2010}
M.~Saffman, T.~G. Walker, and K.~Mølmer, ``Quantum information with {Rydberg}
  atoms,'' \emph{Rev. Mod. Phys.}, vol.~82, no.~3, pp. 2313--2363, Aug. 2010.

\bibitem{RydMag_Art2022}
A.~Artusio-Glimpse, M.~T. Simons, N.~Prajapati, and C.~L. Holloway, ``Modern
  {RF} measurements with hot atoms: A technology review of {Rydberg} atom-based
  radio frequency field sensors,'' \emph{IEEE Microw. Mag.}, vol.~23, no.~5,
  pp. 44--56, May 2022.

\bibitem{liu_space-air-ground_2018}
J.~Liu, Y.~Shi, Z.~M. Fadlullah, and N.~Kato, ``Space-{Air}-{Ground}
  {Integrated} {Network}: {A} {Survey},'' \emph{IEEE Commun. Surv. Tutor.},
  vol.~20, no.~4, pp. 2714--2741, 2018.

\bibitem{Multiband_Sim2020}
M.~S. Sim, Y.-G. Lim, S.~H. Park, L.~Dai, and C.-B. Chae, ``Deep learning-based
  {mmWave} beam selection for {5G} {NR/6G} with sub-6 {GHz} channel
  information: Algorithms and prototype validation,'' \emph{IEEE Access},
  vol.~8, pp. 51\,634--51\,646, 2020.

\bibitem{RydMultiband_Du2022}
Y.~Du, N.~Cong, X.~Wei, X.~Zhang, W.~Luo, J.~He, and R.~Yang, ``Realization of
  multiband communications using different {Rydberg} final states,'' \emph{AIP
  Adv.}, vol.~12, no.~6, p. 065118, Jun. 2022.

\bibitem{RydAMFM_Anderson2021}
D.~A. Anderson, R.~E. Sapiro, and G.~Raithel, ``An atomic receiver for {AM} and
  {FM} radio communication,'' \emph{IEEE Trans. Antennas Propag.}, vol.~69,
  no.~5, pp. 2455--2462, May 2021.

\bibitem{RydPhase_Simons2019}
M.~T. Simons, A.~H. Haddab, J.~A. Gordon, and C.~L. Holloway, ``A {Rydberg}
  atom-based mixer: {Measuring} the phase of a radio frequency wave,''
  \emph{Appl. Phys. Lett.}, vol. 114, no.~11, p. 114101, Mar. 2019.

\bibitem{RydNP_Jing2020}
M.~Jing, Y.~H. Hu, J.~Ma, H.~Zhang, L.~Zhang, L.~Xiao, and S.~Jia, ``Atomic
  superheterodyne receiver based on microwave-dressed {Rydberg} spectroscopy,''
  \emph{Nat. Phys.}, vol.~1, pp. 911--915, Jun. 2020.

\bibitem{RydPhase_And2020}
D.~A. Anderson, R.~E. Sapiro, and G.~Raithel, ``Rydberg atoms for
  radio-frequency communications and sensing: Atomic receivers for pulsed {RF}
  field and phase detection,'' \emph{IEEE Aerosp. Electron. Syst. Mag.},
  vol.~35, no.~4, pp. 48--56, 2020.

\bibitem{RydPhase_Meyer2018}
D.~H. Meyer, K.~C. Cox, F.~K. Fatemi, and P.~D. Kunz, ``Digital communication
  with {Rydberg} atoms and amplitude-modulated microwave fields,'' \emph{Appl.
  Phys. Lett.}, vol. 112, no.~21, p. 211108, May 2018.

\bibitem{RydAOA_Robinson2021}
A.~K. Robinson, N.~Prajapati, D.~Senic, M.~T. Simons, and C.~L. Holloway,
  ``Determining the angle-of-arrival of a radio-frequency source with a
  {Rydberg} atom-based sensor,'' \emph{Appl. Phys. Lett.}, vol. 118, no.~11, p.
  114001, Mar. 2021.

\bibitem{RydMultiband_Meyer2023}
D.~H. Meyer, J.~C. Hill, P.~D. Kunz, and K.~C. Cox, ``Simultaneous multiband
  demodulation using a rydberg atomic sensor,'' \emph{Phys. Review Appl.},
  vol.~19, p. 014025, Jan. 2023.

\bibitem{MIMO_LU2014}
L.~Lu, G.~Y. Li, A.~L. Swindlehurst, A.~Ashikhmin, and R.~Zhang, ``An overview
  of massive {MIMO}: Benefits and challenges,'' \emph{{IEEE} J. Selected Top.
  Signal Process.}, vol.~8, no.~5, pp. 742--758, Oct. 2014.

\bibitem{RIS_Zhang2023}
Z.~Zhang, L.~Dai, X.~Chen, C.~Liu, F.~Yang, R.~Schober, and H.~V. Poor,
  ``Active {RIS} vs. passive {RIS}: Which will prevail in {6G}?'' \emph{IEEE
  Trans. Commun.}, vol.~71, no.~3, pp. 1707--1725, Mar. 2023.

\bibitem{EdgeAI_Mao2017}
Y.~Mao, C.~You, J.~Zhang, K.~Huang, and K.~B. Letaief, ``A survey on mobile
  edge computing: The communication perspective,'' \emph{IEEE Commun. Surv.
  Tutor.}, vol.~19, no.~4, pp. 2322--2358, 2017.

\bibitem{PR_Netrapalli2015}
P.~Netrapalli, P.~Jain, and S.~Sanghavi, ``Phase retrieval using alternating
  minimization,'' \emph{IEEE Trans. Signal Process.}, vol.~63, no.~18, pp.
  4814--4826, Sep. 2015.

\bibitem{AtomicPhysics}
C.~J.~Foot, \emph{Atomic Physics}.\hskip 1em plus 0.5em minus 0.4em\relax
  Oxford University Press, 2005.

\bibitem{QuantumMechanism_Zwiebach2006}
B.~Zwiebach, \emph{Mastering Quantum Mechanics}.\hskip 1em plus 0.5em minus
  0.4em\relax {MIT} Press, 2022.

\bibitem{QuantumOptics_Fox2006}
M.~Fox, \emph{Quantum Optics: {An} Introduction}.\hskip 1em plus 0.5em minus
  0.4em\relax Oxford University Press, 2006.

\bibitem{PR_Dong2023}
J.~Dong, L.~Valzania, A.~Maillard, T.-a. Pham, S.~Gigan, and M.~Unser, ``Phase
  retrieval: {From} computational imaging to machine learning: {A} tutorial,''
  \emph{IEEE Signal Process. Mag.}, vol.~40, no.~1, pp. 45--57, Jan. 2023.

\bibitem{QSN_Bussey2022}
L.~W. Bussey, F.~A. Burton, K.~Bongs, J.~Goldwin, and T.~Whitley, ``Quantum
  shot noise limit in a {Rydberg} {RF} receiver compared to thermal noise limit
  in a conventional receiver,'' \emph{IEEE Sens. Lett.}, vol.~6, no.~9, pp.
  1--4, 2022.

\bibitem{PR_Candes2015}
E.~J. Candès, X.~Li, and M.~Soltanolkotabi, ``Phase retrieval via {Wirtinger}
  flow: Theory and algorithms,'' \emph{{IEEE} Trans. Inf. Theory}, vol.~61,
  no.~4, pp. 1985--2007, Apr. 2015.

\bibitem{Mises_Gattoa2007}
R.~Gattoa and S.~R. Jammalamadaka, ``The generalized von {Mises}
  distribution,'' \emph{Stat. Methodol.}, vol.~4, pp. 341--353, Jul. 2007.

\bibitem{SpecialFunction_Fox2006}
N.~N. Lebedev and R.~A. Silverman, \emph{Special Functions and Their
  Applications}.\hskip 1em plus 0.5em minus 0.4em\relax USA: Courier
  Corporation, 1972.

\bibitem{SIBALIC2017319}
N.~Šibalić, J.~Pritchard, C.~Adams, and K.~Weatherill, ``Arc: An open-source
  library for calculating properties of alkali {Rydberg} atoms,'' \emph{Comput.
  Phys. Commun.}, vol. 220, pp. 319--331, 2017.

\bibitem{CM_Yue2019}
Y.~Li, R.~K. Mallik, and R.~Murch, ``Channel magnitude-based {MIMO} with energy
  detection for {Internet} of {Things} applications,'' \emph{IEEE Internet
  Things J.}, vol.~6, no.~6, pp. 9893--9907, Dec. 2019.

\bibitem{PR_Zhu2023}
J.~Zhu, K.~Liu, Z.~Wan, L.~Dai, T.~J. Cui, and H.~V. Poor, ``Sensing {RISs}:
  {Enabling} dimension-independent {CSI} acquisition for beamforming,''
  \emph{IEEE Trans. Inf. Theory}, vol.~69, no.~6, pp. 3795--3813, Jun. 2023.

\end{thebibliography}


\end{document}